\documentclass[sts]{imsart}

\RequirePackage[OT1]{fontenc}
\usepackage{amsthm,amsmath,natbib}
\RequirePackage[colorlinks,citecolor=blue,urlcolor=blue]{hyperref}

\startlocaldefs
\numberwithin{equation}{section}
\theoremstyle{plain}

\endlocaldefs

\usepackage{enumitem}
\usepackage{subcaption}
\usepackage{multicol}

\usepackage{pgf}
\usepackage{tikz}
\usetikzlibrary{bayesnet}

\usepackage{amsmath,amssymb,xspace,amsthm}

\usepackage{dcolumn}\newcolumntype{d}[1]{D{.}{.}{#1}}

\newtheorem{definition}{Definition}
\newtheorem{idefinition}{Informal Definition}
\newtheorem{theorem}{Theorem}

\def\qedbox{\ifvmode\else\unskip\fi~\penalty10000
    \hfill{\large$\blacksquare$}}

\def\Fig#1{Figure~\ref{#1}}

\def\Thm#1{Theorem~\ref{#1}}

\def\Sec#1{Section~\ref{#1}}
\def\Def#1{Informal~Definition~\ref{#1}}
\def\Appendix#1{Appendix~\ref{#1}}
\def\eref#1{(\ref{#1})}
\def\Eref#1{Eq.~(\ref{#1})}

\def\mm#1{\ensuremath{\boldsymbol{#1}}} 

\def\Stdev{\text{Stdev}}
\def\E{\text{E}}

\def\Prob{\text{Prob}}

\def\RINLA{\texttt{R-INLA}\xspace}

\def\STAN{\texttt{STAN}\xspace}
\def\JAGS{\texttt{JAGS}\xspace}
\def\KLN#1#2{\text{KLD}\left(#1\| #2\right)}
\def\KL#1#2{\int #1 \log\left(\frac{#1}{#2}\right)\;d\mm{x}}

 \newcommand{\norm}[1]{\lVert#1\rVert}

\newcommand{\PCP}{PC prior}

\DeclareMathOperator{\supp}{supp}

\graphicspath{ {figures/} }
\begin{document}

\begin{frontmatter}
\title{Penalising model component complexity: A principled, practical
    approach to constructing priors}
\runtitle{PC Priors}

\begin{aug}
\author{\fnms{Daniel} \snm{Simpson}\thanksref{t0}\ead[label=e1]{D.P.Simpson@warwick.ac.uk}},
\author{\fnms{H\aa{}vard} \snm{Rue}\ead[label=e5]{hrue@math.ntnu.no}},
\author{\fnms{Thiago G.} \snm{Martins}\ead[label=e2]{thigm85@gmail.com}},
\author{\fnms{Andrea} \snm{Riebler}\ead[label=e3]{andrea.riebler@math.ntnu.no}},
\and
\author{\fnms{Sigrunn H.} \snm{S\o{}rbye}\ead[label=e6]{sigrunn.sorbye@uit.no}}

\thankstext{t0}{Corresponding author}

\affiliation{University of Warwick, NTNU, University of Troms\o{} The Arctic University}

\address{Department of Statistics, University of Warwick, Coventry, CV4 7AL, United Kingdom \printead{e1}}
\address{Department of Mathematical Sciences,
    NTNU, Norway\printead{e5,e2,e3}}
    \address{Department of Mathematics and Statistics, University of Troms\o{} The Arctic University \printead{e6}}
%
\end{aug}

\begin{abstract}

    In this paper, we introduce a new concept for constructing prior
    distributions. We exploit the natural nested structure inherent to
    many model components, which defines the model component to be a
    flexible extension of a base model.  Proper priors are defined to
    penalise the complexity induced by deviating from the simpler base
    model and are formulated after the input of a user-defined
    \emph{scaling} parameter for that model component, both in the
    univariate and the multivariate case. These priors are invariant
    to reparameterisations, have a natural connection to Jeffreys'
    priors, are designed to support Occam's razor and seem to have
    excellent robustness properties, all which are highly desirable
    and allow us to use this approach to define default prior
    distributions. Through examples and theoretical results, we
    demonstrate the appropriateness of this approach and how it can be
    applied in various situations.

\end{abstract}




\end{frontmatter}

\bibliographystyle{imsart-nameyear}

%
%
%
%
%
%
%
%


\section{Introduction}
\label{sec1}

The field of Bayesian statistics began life as a  sub-branch of probability theory. From the 1950s onwards, a number of pioneers built upon the Bayesian framework and applied it with great success to real world problems.  The true Bayesian ``moment'', however, began with the advent of Markov chain Monte Carlo (MCMC) methods, which can be viewed as a universal computational device to solve (almost) any problem.  Coupled with user-friendly implementations of MCMC, such as WinBUGS, Bayesian models exploded across fields of research from astronomy to anthropology, linguistics to zoology.  Limited only by their data, their patience, and their imaginations, applied researchers have constructed and applied increasingly complex Bayesian models. The spectacular flexibility of the Bayesian paradigm as an applied modelling toolbox has had a number of important consequences for modern science (see, for example, the special issue of Statistical Science (Volume 29, Number 1, 2014) devoted to Bayesian success stories).

There is a cost to this flexibility. The Bayesian machine is a closed one: given data, an assumed sampling distribution (or likelihood),  a prior probability,  and perhaps a loss function, this machine will spit out a single answer.  The quality of this answer (most usually given in the form of a posterior distribution for some parameters of interest) depends entirely on the quality of the inputs.  Hence, a significant portion of Bayesian data analysis inevitably involves modelling the prior distribution.  Unfortunately, as models become more complex the number of methods for specifying joint priors on the parameter space drops off sharply.

In some cases, the statistician has access to detailed prior information (such as in a sequential experiment) and the choice of a prior distribution is uncontroversial.  At other times, the statistician can construct an informative, subjective or expert prior by interviewing experts in order to elicit
parameters in a parametric family of priors~\citep{kadane1980interactive,book121,art560,book119,book120}.  

More commonly, the prior distributions that are used are not purely
subjective and, as such, are open for criticism.  Non-subjective
priors are used for a number of reasons ranging from a lack of expert
information, to the difficulty in eliciting information about
structural parameters that are further down the model hierarchy, such
as precision or correlation parameters.  Furthermore, as models grow
more complex, the difficulty in specifying expert priors on the
parameters increases.  Martyn Plummer, the author of \JAGS software
for Bayesian inference~\citep{plummer2003jags} goes so far as to
say\footnote{\texttt{http://martynplummer.wordpress.com/2012/09/02/stan/}}
``\emph{[...] nobody can express an informative prior in terms of the
    precision [...]}''.
    
The problem of constructing sensible default priors on model
parameters becomes especially pressing when developing general
software for Bayesian computation.  As developers of the \RINLA (see
 \texttt{http://www.r-inla.org/}) package, which performs
approximate Bayesian inference for latent Gaussian
models~\citep{art451,art500,art522}, we are left with two unpalatable
choices.  We could force the user of \RINLA to explicitly define a joint prior
distribution for all parameters in the model.  Arguably, this is the
correct thing to do, however, the sea of confusion around how to
properly prescribe priors makes this undesirable in practice.  The
second option is to provide default priors. These, as currently
implemented, are selected partly out of the blue with the hope that
they lead to sensible results.

This paper is our attempt to provide a broad, useful framework for building priors for a large class of hierarchical models.  The priors  we develop, which we call
\emph{Penalised Complexity} or \PCP s, are informative priors.  The
information in these priors is specified in terms of four underlying
principles.  This has a twofold purpose. The first purpose is to
communicate the exact information that is encoded in the prior in
order to make the prior interpretable and easier to elicit.     \PCP s
have a single parameter that the user must set, which controls the
amount of flexibility allowed in the model.
This parameter can be set using ``weak'' information that is
frequently available \citep{art547}, or by appealing to some other subjective criterion such as ``calibration'' under some assumptions about future experiments \citep{draper2006coherence}.  This is in line with  Draper's  call for ``transparent subjectivity''.  

Following in the footsteps of Lucien Le Cam \citep[``Basic Principle 0. Do not trust any principle.''][]{le1990maximum} and (allegedly) Groucho Marx (``Those are my principles, and if you don't like them\ldots well, I have others.''),   the second purpose of building
\PCP s from a set of principles is to allow us to change these
principles when needed. For example, in Sections \ref{sec:sparsity}
and \ref{sec:ext} we modify single principles for, respectively,
modelling and computational reasons.  This gives the \PCP\ framework
the advantage of flexibility without sacrificing its simple structure.  We stress that the principles provided in this paper do not absolve modellers of the responsibility to check their suitability \citep[see, for example, ][who argued that the principles underlying the reference prior approach are inappropriate for spatial data]{palacios2006non}.

We believe that \PCP s are general enough to be used in realistically
complex statistical models and are straightforward enough to be used
by general practitioners. Using only weak information, \PCP s
represent a unified prior specification with a clear meaning and
interpretation.  The underlying principles are designed so that
desirable properties follow automatically: invariance regarding
reparameterisations, connection to Jeffreys' prior, support of Occam's
razor principle, and empirical robustness to the choice of the
flexibility parameter. Our approach is not restricted to any specific
computational method as it is a principled approach to prior
construction and therefore relevant to any application involving
Bayesian analysis.

We do not claim that the priors we propose are optimal or unique, nor do we claim that the principles are universal.  Instead, we make the more modest claim that these priors are useful, understandable, conservative, and better than doing nothing at all.  

\subsection{The models considered in this paper}
While the goals of this paper are rather ambitious, we will necessarily restrict ourselves to a specific class of hierarchical model, namely additive models.  The models we consider have a non-trivial unobserved latent structure.  This latent structure is made up of a number of \emph{model components}, the structure of which is controlled by a small number of \emph{flexibility parameters}.  We are interested in latent structures in which each model component is  added for modelling purposes. We do not focus on the case where the hierarchical structure is added to increase the robustness of the model \citep[See Chapter 10 of ][for a discussion types of hierarchical structure]{robert2007bayesian}.  This \emph{additive model} viewpoint is the key to understanding many of the choices we make, in particular the concept of the ``base model'', which is covered in detail in \Sec{sec:base}.

An example of the type of model we are considering is the spatial survival model proposed by \citet{art458}, where the log-hazard rate is modelled according to a Cox proportional hazard model as $$
\log(\text{hazard}_j) = \log(\text{baseline}) + \beta_0 + \sum_{i=1}^p \beta_i x_{ij} + u_{r_j} + v_j,
$$
where $x_{ij}$ is the $i$th covariate for case $j$, $u_{r_j}$ is the value of the spatial structured random effect for the region $r_j$ in which case $j$ occurred, and $v_j$ is the subject specific log-frailty.  Let us focus on the \emph{model component} $ \mm{u} \sim N(0, \tau^{-1} \mm{Q}^{-1})$, where $\mm{Q}$ is is the structure matrix of the first order CAR  model on the regions and $\tau$ is an inverse scaling parameter \citep[Chapter 3]{book80}. The model component $\mm{u}$ has one flexibility parameter $\tau$, which controls the scaling of the structured random effect. The other model components are $\mm{v}$ and $\mm{\beta}$, which have one and zero  (assuming a uniform prior on $\mm{\beta}$) flexibility parameters respectively. We will consider this case in detail in \Sec{sec:bym}.

 We emphasise that we are not interested in the case where the number of flexibility parameters grows as we enter an asymptotic regime (here as the number of cases increases). The only time we consider models where the number of parameters grows in an ``asymptotic'' way is \Sec{sec:sparsity}, where we consider sparse linear models. In that section we discuss a (possibly necessary) modification to the prior specification given below (specifically Principle 3 in \Sec{sec2}). We also do not consider models with discrete components.

\subsection{Outline of the paper}

The plan of the paper is as follows. \Sec{sec:review} contains an overview of common techniques for setting priors for hierarchical models. In~\Sec{sec2} we will define our
principled approach to design priors and discuss its connection to the
Jeffreys' prior. In~\Sec{sec3}, we will study properties of \PCP s
near the base model and its behaviour in a Bayesian hypothesis testing
setting. Further, we provide explicit risk results in a simple
hierarchical model and discuss the connection to sparsity
priors. 
In~\Sec{sec:t} we study prior choices for the degrees of
freedom parameter in the Student-t distribution and perform an
extensive simulation study to investigate properties of the new prior.
 In~\Sec{sec:bym} we
discuss the BYM-model for disease mapping with a possible smooth
effect of an ecological covariate, and we suggest a new
parameterisation of the model in order to facilitate improved control
and interpretation.  
\Sec{sec:mvprobit} extends the method to multivariate parameters and
we derive  principled priors for correlation matrices in the
context of the multivariate probit model. \Sec{sec:ext} contains a
discussion of how to extend the framework of \PCP s to hierarchical
models by defining joint \PCP s over model components that take the
model structure into account.  This technique is demonstrated on an
additive logistic regression model.  We end with a discussion
in~\Sec{sec:disc}. The Appendices host technical details and
additional results.

\section{A  guided tour of non-subjective priors for Bayesian hierarchical models} \label{sec:review}

The aim of this section is to review the existing methods for setting non-subjective priors for parameters in Bayesian hierarchical models. We begin by discussing, in order of descending purity, objective priors, weakly-informative priors, and what we call risk-averse priors.  We then consider a special class of priors that are important for hierarchical models, namely priors that encode some notion of a base model.   Finally, we investigate the main concepts that we feel are most important for setting priors for parameters in hierarchical models, and we look at related ideas in the literature. 

In order to control the size of this section, we have made two major decisions.  The first is that we are focussing exclusively on methods of prior specification that could conceivably be used in all of the examples in this paper. The second is that we focus entirely on priors for prediction. It is commonly \citep[although not exclusively][]{bernardo2011integrated,rousseaumoment,kamary2014testing} held that we need to use different priors for testing than those used for prediction.  We return to this point in \Sec{sec:hypothesis_testing}.  A discussion of alternative priors for the specific examples in this paper is provided in the relevant section.   We also do not consider data-dependent priors and empirical Bayes procedures.

\subsection{Objective priors}
The concept of prior specification furthest
from expert elicitation priors is that of ``objective'' priors
\citep{art534,art555,art558,art557,kass1996selection}.  These aim to inject as little
information as possible into the inference procedure.  Objective
priors are often strongly design-dependent and are not uniformly accepted
amongst Bayesians on philosophical grounds; see for example discussion
contributions to \citet{art555} and \citet{art556}, but they are  useful 
(and used) in practice. The most common constructs in this
family are Jeffreys' non-informative priors and their extension
``reference priors''~\citep{art558}.  These priors are often
improper and great care is required to ensure that the resulting
posteriors are proper.  If chosen carefully, the use of
non-informative priors will lead to appropriate parameter estimates as
demonstrated in several applications by \cite{kamary2014}.  It can be also shown theoretically that for sufficiently nice models, the posterior resulting from a reference prior analysis matches the results of classical maximum likelihood estimation to second order \citep{reid2003some}.  

While reference priors have been successfully used for classical models, they have a less triumphant history for hierarchical models. There are two reasons for this.  The first reason is that reference priors are model dependent and notoriously difficult to derive for complicated models.  Furthermore, when the model changes, such as when a model component is added or removed, this prior needs to be re-computed. This
does not mesh well with the practice of applied statistics, in which a
``building block'' approach is commonly used and several slightly
different model formulations will be fitted simultaneously to the same
dataset.  The second reason that it is difficult to apply the reference prior methodology to hierarchical models is that reference priors depend  on an ordering of the parameters.  In some applications, there may be a natural ordering.  However, this can also lead to farcical situations.  Consider the problem of determining the $m$ parameters of a multinomial distribution (c.f. the $3$ weights $w_i$ in \Sec{sec:ext}).  In order to use a reference prior, one would need to choose one of the $m!$ orderings of the  allocation probabilities. This is not a benign choice as the the joint posterior will depend strongly on the ordering. \citet{berger2015overall} proposed some work arounds for this problem, however it is not clear that they enjoy the good properties of pure reference priors and it is not clear how to apply these to even moderately complex hierarchical models \citep[see the comment of ][]{rousseau2015comment}.   In spite of
these shortcomings, the reference prior framework is the only complete
framework for specifying prior distributions.

\subsection{Weakly informative priors}
Between objective and expert priors lies the realm of ``weakly
informative'' priors
\citep{art547,art548,evans2011weak,polson2012half}.  These priors are
constructed by recognising that while you usually do not have strong
prior information about the value of a parameter, it is rare to be
completely ignorant.  For example, when estimating the height and
weight of an adult, it is sensible to select a prior that gives mass
neither to people who are five metres tall, nor to those who only
weigh two kilograms.  This use of weak prior knowledge is often sufficient
to regularise the extreme inferences that can be obtained using
maximum likelihood \citep{le1990maximum} or non-informative priors.  To date, there has been
no attempt to construct a general method for specifying weakly
informative priors.  

Some known weakly informative priors, like the half-Cauchy distribution on the standard deviation of a normal distribution, can lead to well better predictive inference than reference priors.  To see this, we consider the problem of inferring the mean of a normal distribution with known variance.  The posterior mean obtained from the reference prior is the sample mean of the observations.   Charles Stein  and coauthors showed that this estimator is inadmissible in the sense that you can find an estimator that is better than it for any value of the true parameter \citep{james1961estimation,stein1981estimation}.  This problem fits within our framework if we consider the mean to be \emph{a priori} normally distributed with an unknown standard deviation, which is a flexibility parameter.  For the half-Cauchy prior on the standard deviation, \citet{polson2012half} showed that the resulting posterior mean dominates the sample mean and has lower risk than the Stein estimator for small signals.  This shows that the right sort of weak information has the potential to greatly improve the quality of the inference.  This case is explored in greater detail in \Sec{section:risk}.    There are no general theoretical results that show how to build priors with good risk properties for the broader class of models we are interested in, but the intuition is that weakly informative priors can strike a balance between fidelity to a strong signal, and shrinkage of a weak signal. We interpret this as the prior on the flexibility parameter (the standard deviation) allowing extra model complexity, but not forcing it.

\subsection{\emph{Ad hoc}, risk averse, and computationally convenient prior specification}
The methods of prior specification we have considered in the previous two sections are distinguished by the presence of an (abstract) ``expert''.  For objective priors, the expert is demanding the priors encode only minimal information, whereas for weakly informative priors, the expert is choosing a specific type of ``weak'' information that the prior will encode.  The priors that we consider in this section lack such an expert and, as such, are difficult to advocate for.  These methods, however, are used in a huge amount of applied Bayesian analysis.

The most common non-subjective approach to prior specification for hierarchical models is to use a prior that has been previously used in the literature for a similar problem.  This \emph{ad hoc} approach can be viewed as a \emph{risk averse} strategy, in which the choice of prior has been delegated to another researcher. In the best cases, the
chosen prior was originally selected in a careful, problem independent
manner for a similar problem to the one the statistician is solving \citep[for example, the priors in ][]{gelman2013bayesian}.
More commonly, these priors have been carefully chosen for the problem
they were designed to solve and are inappropriate for the new
application \citep[such as the priors in][]{muff-etal-2015}.  The lack of a dedicated ``expert'' guiding these prior choices can lead to troubling inference. Worse still is the idea that, as the prior was selected from the literature or is in common use, there is some sort of justification for it.

Other priors in the literature have been selected for purely computational reasons.  The main example of these priors 
are conjugate priors for exponential families \citep[Section 3.3]{robert2007bayesian}.  These priors ensure the 
existence of analytic conditional distributions, which allow for the easy implementation of a Gibbs sampler.  While Gibbs samplers are an important part of the historical development of Bayesian statistics, we tend to favour sampling methods based on the joint distribution, such as \STAN, as they tend to perform better.  Hence, we only require that the priors have a tractable density and that it is straightforward to explore the parameter space.  This latter requirement is particularly important when dealing with structured multivariate parameters, for example correlation matrices, where it may not be easy to move around the set of valid parameters (see \Sec{sec:mvprobit} for a discussion of this).

Some priors from the literature are not
sensible.  An extreme example of this is the enduring popularity of
the $\Gamma(\epsilon, \epsilon)$ prior, with a small $\epsilon$, for
inverse variance (precision) parameters, which has been the
``default''\footnote{We note that this recommendation has been revised, however these priors are still widely used in the literature.} choice in the WinBugs~\citep{tech23} example
manuals. However, this prior is well known to be a choice with severe
problems; see the discussion in \citet{art460} and \citet{book113}.
Another example of a bad ``vague-but-proper'' prior is a uniform prior
on a fixed interval for the degrees of freedom parameter in a
Student-t distribution.  The results in Section \ref{sec:t} show that these
priors, which are also used in the WinBugs manual, get increasingly
informative as the upper bound increases.

One of the unintentional consequences of using risk averse priors is that they will usually lead to independent priors on each of the hyperparameters. For complicated models that are overspecified or partially identifiable, we do not think this is necessarily a good idea, as we need some sort of shrinkage or sparsity to limit the flexibility of the model and avoid over-fitting.  The tendency towards over-fitting is a property of the full model and independent priors on the components may not be enough to mitigate it.
\citet{fuglstad2015interpretable} define a joint prior on the partially identifiable hyperparameters for a Gaussian random field that is re-parameterised to yield independent priors, whereas  \citet{bhattacharya2014dirichlet} define a sparsity prior for an over-specified linear model by considering a prior that concentrates on the boundary of a simplex, that in specific cases factors into independent priors on each model component. This strongly suggests that for useful inference in hierarchical models, it is important to consider the \emph{joint} effect of the priors and it is not enough to simply focus on the marginal properties. 

While the tone of this section has been quite negative, we do not wish to give the impression that all inference obtained using risk averse or computationally convenient priors will not be meaningful.  We only want to point out  that a lot of work needs to be put into checking the suitability of the prior for the particular application before it is used.  Furthermore, the suitability (or not) or a specific joint prior specification will depend in subtle and complicated ways on the global model specification.  An interesting, but computationally intensive, method for reasserting the role of an ``expert'' into a class of \emph{ad hoc} priors, and hence recovering the justifiable status of objective and weakly-informative priors, is the method of \emph{calibrated Bayes} \citep{rubin1984bayesianly}.  In calibrated Bayes, as practiced with aplomb by \citet{browne2006comparison},   a specific prior distribution is chosen from a possibly \emph{ad hoc} class by ensuring that, under correct model specification, the credible sets are honest credible regions.

\subsection{Priors specified using a base model} \label{sec:base}  

One of the key challenges when building a prior for a hierarchical model is finding a way to control against over-fitting. In this section, we consider a number of priors that have been proposed in the literature that are linked through the abstract concept of  a ``base model''.

\begin{definition} \label{defn:base_model}
For a model component with density $\pi(\mm{x} \mid \mm{\xi})$ controlled by a flexibility parameter $\mm{\xi}$, the base model is the ``simplest'' model in the class.  For notational clarity, we will take this to be the model corresponding to $\mm{\xi}=0$.  It will be common for $\mm{\xi}$ to be non-negative.   The flexibility parameter is often a scalar, or a number of independent scalars, but it can also be a vector-valued parameter.
\end{definition}

 This allows us to interpret
$\pi(\mm{x} \mid \mm{\xi})$ as a \emph{flexible extension} of the base model,
where increasing values of $\mm{\xi}$ imply increasing deviations from the
base model. The idea of a base model is reminiscent of a ``null hypothesis'' and thinking of what a sensible hypothesis to test for $\mm{\xi}$ is a good way to specify a base model.  We emphasise, however, that we are not using this model to do testing, but rather to control flexibility and reduce over-fitting thereby improving predictive performance.

The idea of a base models is strongly linked to the ``additive'' version of model building.  As we add more model components, it is necessary to mitigate the increase in flexibility by placing a ``barrier'' to their application.  This suggests that we should shrink towards the simplest version of the model component, which often corresponds to it not being present.

A few simple examples will fix the idea.
\begin{description}
\item[Gaussian random effects] Let $\mm{x}\mid\xi$ be a multivariate
    Gaussian with zero mean and precision matrix $\tau \mm{I}$ where
    $\tau=\xi^{-2}$. Here, the base model puts all the mass at $\xi=0$,
    which is appropriate for random effects where the natural
    reference is absence of these effects. In the multivariate case
    and conditional on $\tau$, we can allow for correlation among the
    model components where the uncorrelated case is the base model.
\item[Spline model] Let $\mm{x}\mid\xi$ represent a spline model with
    smoothing parameter $\tau = 1/\xi$. The base model is the infinite
    smoothed spline which can be a constant or a straight line,
    depending on the order of the spline or in general its
    null-space. This interpretation is natural when the spline model
    is used as a flexible extension to a constant 
     or in generalised additive models, which
    can (and should) be viewed as a flexible extension of a
    generalised linear model.
\item[Time dependence] Let $\mm{x}\mid\xi$ denote an auto-regressive
    model of order 1, unit variance and lag-one correlation
    $\rho=\xi$. Depending on the application, then either $\rho=0$ (no
    dependence in time) or the limit $\rho=1$ (no changes in time) is
    the appropriate base model.
\end{description}

The base model  primarily finds a home in the idea of ``spike-and-slab'' priors \citep{george1993variable,ishwaran2005spike}.  These models specify a prior on $\mm{\xi}$ as a mixture of a point mass at the base model and a diffuse absolutely continuous prior over the remainder of the parameter space.  These priors successfully control over-fitting and simultaneously perform prediction and model selection.  The downside is that they are computationally unpleasant and specialised tools need to be built to infer these models. Furthermore, as the number of flexibility parameters increases, exploring the entire posterior quickly becomes infeasible.

Spike-and-slab models are particularly prevalent in high-dimensional regression problems, where a more computationally pleasant alternative has been proposed.  These local-global priors are scale mixtures of normal distributions, where the priors on the variance parameters are chosen to mimic the spike-and-slab structure \citep{bhattacharya2014dirichlet}. While these parameters are \emph{not}, by our definition, flexibility parameters, this suggests that non-atomic priors can have the advantages of spike-and-slab priors without the computational burden.

In order to further consider what a non-atomic prior must look like to take advantage of the base model, we consider the following Informal Definition. 
\begin{idefinition} \label{def:overfit} A prior $\pi(\mm{\xi})$ forces
    overfitting (or \emph{overfits}) if the density of the prior is
    zero at the base model. A prior avoids overfitting if its density is decreasing with a maximum at  the base model.
\end{idefinition}
A prior that overfits will drag the posterior towards the more
flexible model and the base model will have almost no support in the
posterior, even in the case where the base model is the true
model. Hence, when using an overfitting prior, we are unable to
distinguish between flexible models that are supported by the data and
flexible models that are a consequence of the prior choice.  For variance components, this idea is implicit in the work of \citet{art547,art548},  while for high-dimensional models, it is strongly present in the concept of the local-global scale mixture of normals.  

At first glance, \Def{def:overfit} appears to be overly harsh and it requires justification.  We know that the minimal requirement for the posterior to contract to the base model (under repeated observations from the base model) is that the prior puts enough mass in small neighbourhoods around $\xi=0$.  The shape of these neighbourhoods depends  on the model, which makes this condition difficult to check for the models we are considering, however priors that don't overfit in the sense of \Def{def:overfit} are likely to fulfil it. The second justification is that, in order to avoid overfitting, we want the model to require stronger information to ``select'' a flexible model than it does to ``select'' the base model.  \Def{def:overfit} ensures that the probability in a small ball around $\xi = 0$  is larger than the probability in the corresponding ball around any other $\xi_0 \neq 0$.  This can be seen as an absolutely continuous generalisation of the spike-and-slab concept.

While the careful specification of the prior near the base model controls against overfitting, it is also necessary to consider the tail behaviour. While it is not the focus of this paper, this is best understood in the high-dimensional regression case.  The spike-and-slab approach typically uses a slab distribution with exponential tails such as the Laplace distribution \citep{castillo2012needles}, while global-local scale mixtures of normals need heavier tails in order to simultaneously control the probability of a zero and the behaviour for large signals (see Theorem \ref{thm:general_sparse}, although the idea is implicit in \citet{bhattacharya2014dirichlet} and the proof of Proposition 1 in \citet{castillo2014bayesian}).

There is a second argument for heavy tails put forth by  \citet{art547,art548}, who suggested that they are necessary for robustness.  An excellent review of the role complementary role of priors and likelihoods for Bayesian robustness can be found in \citet{o2012bayesian}.  From this, we can see that this theory is only developed for inferring location-scale families or natural parameters in exponential families and is unclear how to generalise these results to the types of models we are considering. 
 In our experiments, we have found little difference between half-Cauchy and exponential tails, whereas we found huge differences between exponential and Gaussian tails (which performed badly when the truth was a moderately flexible model).

\subsection{Methods for setting joint priors on flexibility parameters in hierarchical models: some desiderata } \label{sec:desiderata} 
We conclude this tour of prior specifications by detailing what we look for in a joint prior for parameters in a hierarchical model and pointing out priors that have been successful in fulfilling at least some of these.  This list is, of course, quite personal, but we believe that it is  broadly sensible.  We wish to emphasise that the desiderata listed below only make sense in the context of hierarchical models with \emph{multiple model components} and it does not make sense to apply them to less complex models.    The remainder of the paper can be seen as our attempt to construct a system for specifying priors that at least partially consistent with this list.

\textbf{D1: The prior should not be non-informative} Hierarchical models are frequently quite involved and, even if it was possible to compute a non-informative prior, we are not  convinced it would be a good idea. Our primary concern is the stability of inference.  In particular, if a model is over-specified, these priors are likely to lead to badly over-fitted posteriors.  Outside the realm of formally non-informative priors, emphasising ``flatness'' can lead to extremely prior-sensitive inferences \citep{art547}.  This  should not be interpreted as us calling for massively informative priors, but rather a recognition that for complex models, a certain amount of extra information needs to be injected to make useful inferences.

\textbf{D2: The prior should be aware of the model structure}  Roughly speaking, we want to ensure that if a subset of the  parameters control a single aspect of the model, the prior on these parameters is set jointly.  This also suggests using a parameterisation of the model that, as much as possible, has parameters that only control one aspect of the model. An example of this is the local-global method for building shinkage priors. Consider the normal means model $y_i \sim N(\theta_i,1)$, $\theta_i \sim N(0, \sigma^2 \psi_i^2)$, where the prior on both the global variance parameter $\sigma^2$ and the local variance parameter $\psi_i^2$ have support on $[0, \infty)$. \citet{bhattacharya2014dirichlet} suggested a sensible re-parameterisation, where the local parameters $\{\psi_i\}$ are constrained to lie on a simplex.  With this restriction, the parameters can be interpreted as an overall variance $\sigma^2$ while the $\psi_i$s control how that variance is distributed to each of the model components.  This decoupling of the parameters makes it considerably easier to set sensible prior distributions.

 While this type of decoupling is highly desirable for setting priors, it is not always clear how to achieve it. To see this, consider a negative-binomial GLM, where the (implied) prior data-level variance is controlled by both by the (implied) prior on the log mean and the prior on the overdispersion parameter.  Furthermore, if the overdispersion parameter is of scientific interest, this type of decoupling would be a terrible idea!

\textbf{D3: When re-using the prior for a different analysis, changes in the problem should be reflected in the prior} A prior specification should be explicit about what needs to be changed when applying it to a similar but different problem.  An easy example is that the prior on the scaling parameter of a spline model needs to depend on the number of knots \citep{art521}. A less clear cut example occurs when inferring a deep multi-level model.  In this case, the data is less informative about variance components further down the hierarchy, and hence it may be desirable to specify a stronger prior on these model components.

\textbf{D4: The prior should limit the flexibility of an over-specified model} This desideratum is related to the discussion in \Sec{sec:base}. It is unlikely that priors that do not have good shrinkage properties will lead to good inference for hierarchical models.

\textbf{D5: Restrictions of the prior to identifiable submanifolds of the parameter space should be sensible} As more data appears, the posterior will contract to a submanifold of the parameter space. For an identifiable model, this submanifold will be a point, however any number of models that appear in practice are only partially identifiable \citep{gustafson2005model}. In this case, the posterior limits to the restriction of the prior along the identifiable submanifold and it is reasonable to demand that the prior is sensible along this submanifold. A case where it is not desirable to have a non-informative prior on this submanifold is given in \citet{fuglstad2015interpretable}.

\textbf{D6: The prior should control what a parameter does, rather than its numerical value}
A sensible method for setting priors should (at least locally) indifferent to the parameterisation used.  It does not make sense, for example, for the posterior to depend on whether the modeller prefers working with the standard deviation, the variance or the precision of a Gaussian random effect. 

The idea of using the distance
between two models as a reasonable scale to think about priors dates
back to \citet{jeffreys1946invariant} pioneering work to obtain priors
that are invariant to reparameterisation. \citet{art523} build on the
early ideas of \citet{book112} to derive objective priors for
computing Bayes factors for Bayesian hypothesis tests; see also
\citet[Sec.~6.4]{art554}.  They use divergence measures between the
competing models to derive the required proper priors, and call those
derived priors divergence-based (DB) priors. Given the prior
distribution on the parameter space of a full encompassing model,
\citet{consonni2008} used Kullback-Leibler projection, in the context
of Bayesian model comparison, to derive suitable prior distributions
for candidate submodels.

\textbf{D7: The prior should be computationally feasible} If our aim is to perform applied inference, we need to ensure the models allow for computations within our computational budget.  This will always lead to a very delicate trade-off that needs to be evaluated for each problem.  However, the discussion in \Sec{sec:base} suggests that when we want to include information about a base model (see D4), we can either use a computationally expensive spike-and-slab approach or a cheaper scale-mixture of normals approach.

\textbf{D8: The prior should be consistent with theory}
This is the most difficult desideratum to fulfil.  Ideally, we would like to ensure that, for some appropriate quantities of interest, the estimators produced using these priors have theoretical guarantees.  It could be that we desire optimal posterior contraction, asymptotic normality, good predictive performance under mis-specification, robustness against outliers, admissibility in the Stein sense or any other ``objective'' property.  At the present time, there is essentially no theory of any of these desirable features  for the types of models that we are considering in this paper.  As this gap in the literature closes, it will be necessary to update recommendations on how to set a prior for a hierarchical model to make them consistent with this new theory.



\section{Penalised complexity priors}
\label{sec2}

In this section we will outline our approach to constructing 
\emph{penalised complexity priors} (\PCP s) for a univariate
parameter, postponing the extensions to the multivariate case to \Sec{sec:multivar_PC}. These priors, which are fleshed out in further sections, satisfy most of the desiderata listed in \Sec{sec:desiderata}.  We demonstrate these principles by deriving the \PCP\ for the precision of a Gaussian random effect. 

\subsection{A principled definition of the
    \PCP} \label{sec:principles}

We will now state and discuss our principles for constructing a prior
distribution for $\xi$.
\begin{enumerate}[label=\bfseries Principle \arabic*:]
\item \textbf{Occam's razor.} We invoke the principle of parsimony,
    for which simpler model formulations should be preferred until
    there is enough support for a more complex model. Our simpler
    model is the base model hence we want the prior to penalise
    deviations from it. From the prior alone we should prefer the
    simpler model and the prior should be decaying as a function of a
    measure of the increased complexity between the more flexible
    model and the base model.
\item \textbf{Measure of complexity.} We will use the Kullback-Leibler
    divergence (KLD) to measure the increased complexity
    \citep{kld1951}.  Between densities $f$ and $g$, the KLD is
    defined by
    \begin{displaymath}
        \KLN{f}{g} = \KL{f(\mm{x})}{g(\mm{x})}. 
    \end{displaymath}
    KLD is a measure of the information lost when the base model $g$
    is used to approximate the more flexible model $f$. The asymmetry
    in this definition is vital to the success of this method: Occam's
    razor is naturally non-symmetric and, therefore, our measure of
    complexity must also be non-symmetric. In order to use the KLD, it
    is, however, necessary to transform it onto a physically
    interpretable ``distance'' scale.  Throughout this paper, we will
    use the (unidirectional) measure $d(f || g) = \sqrt{2\KLN{f}{g}}$,
    to define the distance between the two models. (We will use this
    phrase even though it is not a metric.)  Hence, we consider $d$ to
    be a measure of complexity of the model $f$ when compared to model
    $g$.  Here, the factor ``2'' is introduced for convenience and the
    square root deals with the power of two that is naturally
    associated with the KLD (see the simple example at the end of this
    subsection).
\item \textbf{Constant rate penalisation.} Penalising the deviation
    from the base model parameterised with the distance $d$, we use a
    constant decay-rate $r$, so that the prior satisfies
    \begin{displaymath}
        \frac{\pi_d(d + \delta)}{\pi_d(d)} =  r^{\delta}, \qquad d,\delta \ge 0
    \end{displaymath}
    for some constant $0 < r< 1$.  This will ensure that the relative
    prior change by an extra $\delta$ does not depend on $d$, which is
    a reasonable choice without extra knowledge (see the discussion on 
    tail behaviour in \Sec{sec:base}). Deviating from the
    constant rate penalisation implies to assign different decay rates
    to different areas of the parameter space. However, this will
    require a concrete understanding of the distance scale for a
    particular problem, see Section~\ref{sec:sparsity}.  Further, the
    mode of the prior is at $d=0$, i.e.\ the base model.  The constant
    rate penalisation assumption implies an exponential prior on the
    distance scale, $\pi(d)= \lambda \exp(-\lambda d)$, for $r =
    \exp(-\lambda)$.  This corresponds to the following prior on the
    original space
    \begin{equation}
        \pi(\xi)  =\lambda \mathrm{e}^{-\lambda d(\xi)}
        \left|\frac{\partial d(\xi)}{\partial \xi}\right|.
    \end{equation}
    In some cases $d$ is upper bounded and we use a truncated
    exponential as the prior for $d$.
\item \textbf{User-defined scaling.} The final principle needed to
    completely define a \PCP\ is that the \emph{user} has an idea of a
    sensible size for the parameter of interest or a property of the corresponding
    model component. This is similar to the principle behind weakly
    informative priors.  In this context, we can select $\lambda$ by
    controlling the prior mass in the tail. This condition is of the
    form
    \begin{equation}
        \label{eqn:weak_bound}
        \Prob(Q(\xi) > U) = \alpha,
    \end{equation}
    where $Q(\xi)$ is an interpretable transformation of the
    flexibility parameter, $U$ is a ``sensible'', user-defined upper
    bound that specifies what we think of as a ``tail event'', and
    $\alpha$ is the weight we put on this event.  This condition
    allows the user to prescribe how informative the resulting \PCP\
    is.
\end{enumerate}
The \PCP\ procedure is invariant to reparameterisation, since the
prior is defined on the distance $d$, which is then transformed to the
corresponding prior for $\xi$. This is a major advantage of \PCP s,
since we can construct the prior without taking the specific
parameterisation into account.

The \PCP\ construction is consistent with the desiderata listed in \Sec{sec:desiderata}.  Limited flexibility (D4), controlling the effect rather than the value (D6), and informativeness  (D1) follow from Principles 1, 2, and 4 respectively.  Lacking more detailed theory for hierarchical models, Principle 3 is consistent with existing theory (D8).   We argue that computational feasibility (D7) follows from the ``computationally aware'' \Def{def:overfit}, which favours absolutely continuous priors over spike-and-slab priors.  Building ``structurally aware'' priors (D2) is discussed in Sections \ref{sec:bym} and \ref{sec:ext}. The idea that a prior should change in an appropriate way when the model changes is discussed in \Sec{sec:bym}.  The desideratum that the prior is meaningful on identifiable submanifolds (D5) is discussed in \citet{fuglstad2015interpretable}.

\subsection{The \PCP\ for the precision of a Gaussian random effect}
\label{sec:pc.prec}

The classical notion of a random effect has proven to be a convenient
way to introduce association and unobserved heterogeneity.  We will
now derive the \PCP\ for the precision parameter $\tau$ for a Gaussian
random effect $\mm{x}$, where $\mm{x}\sim{\mathcal N}(\mm{0},
\tau^{-1}\mm{R}^{-1})$, $\mm{R}\succeq 0$.  The \PCP\ turns out to be
independent of the actual choice of $\mm{R}$, including the case where
$\mm{R}$ is not of full rank, like in popular intrinsic models, spline
and thin-plate spline models~\citep{book80}. The natural base model is
the absence of random effects, which corresponds to
$\tau=\infty$. In the rank deficient case, the natural base model is that the 
effect belongs to the nullspace of $\mm{R}$, which also corresponds to $\tau=\infty$.
This base model leads to a useful negative result.
\begin{theorem}
    \label{theorem:2}
    Let $\pi_{\tau}(\tau)$ be an absolutely continuous prior for $\tau > 0$ where
    $\E(\tau)<\infty$, then $\pi_d(0) = 0$ and the prior overfits (in the sense of  \Def{def:overfit}).
\end{theorem}
The proof is given in \Appendix{sec:proof}. Note that all commonly
used $\Gamma(a,b)$ priors with expected value $a/b < \infty$ will
overfit. \citet{art562} and \citet{col37} demonstrate overfitting due
to Gamma priors and suggest using a (half) Gaussian prior for the
standard deviation to overcome this problem, as suggested by
\citet{art547}; See also \citet{roos-held-2011} and the discussion of
\citet{lunn2009}.

The \PCP\ for $\tau$ is derived in
\Appendix{appendix:normal} as a type-2 Gumbel distribution \begin{equation}\label{eq:prec}%
    \pi(\tau) = \frac{\lambda}{2} \tau^{-3/2} \exp\left(-\lambda
      \tau^{-1/2}\right), \qquad \tau > 0.
\end{equation} The
density is given in~\Eref{eq:prec} and has no integer moments. This
prior also corresponds to an exponential distribution with rate $\lambda$ for the
standard deviation. The parameter $\lambda$ determines the magnitude
of the penalty for deviating from the base model and higher values
increase this penalty. As previously, we can determine $\lambda$ by
imposing a notion of \emph{scale} on the random effects. This requires
the user to specify $(U, \alpha)$ so that $\text{Prob}(1/\sqrt{\tau} >
U)=\alpha$. This implies that $\lambda = -\ln(\alpha)/U$. As a rule of
thumb, the marginal standard deviation of $\mm{x}$ with
$\mm{R}=\mm{I}$, after the type-2 Gumbel distribution for $\tau$ is
integrated out, is about $0.31U$ when $\alpha=0.01$. This means that
the choice $(U=0.968, \alpha=0.01)$ gives $\Stdev(\mm{x}) \approx
0.3$. The interpretation of the marginal standard deviation of a
random effect is more direct and intuitive than choosing
hyperparameters of a given prior.

The new prior is displayed in~\Fig{fig:prec} for $(U=0.3/0.31,
\alpha=0.01)$, together with the popular $\Gamma(1,b)$ prior, where
the shape is $1$ and rate is $b$.  In applications, the ``art'' is to
select an appropriate value for $b$.  We selected $b$ so that the
marginal variance for the random effects are equal for the two priors.
Panel~(a) shows the two priors on the precision scale and panel~(b)
shows the two priors on the distance scale.  The priors for low
precisions are quite different, and so are the tail behaviours. For
large $\tau$, the new prior behaves like $\tau^{-3/2}$, whereas the
Gamma prior goes like $\exp(-b\tau)$. This is a direct consequence of
the importance the new prior gives to the base model, i.e.\ the
absence of random effects. Panel~(b) demonstrates that the Gamma prior
has density zero at distance zero, and hence, does not prevent
over-fitting.

\begin{figure}[tbp]

    \begin{subfigure}{0.45\textwidth}
    \centering
    \includegraphics[width=0.9\textwidth,angle=270]{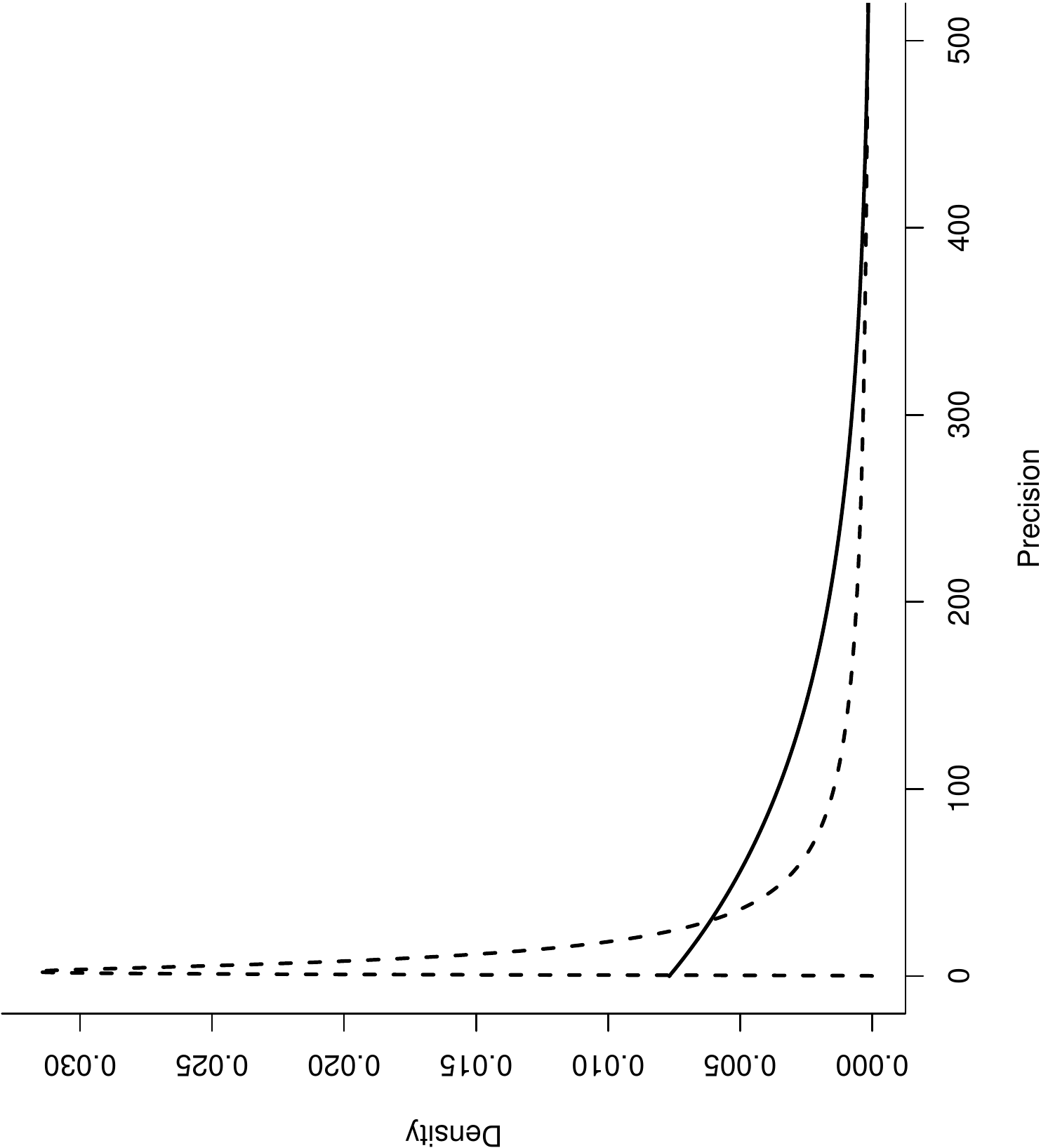}
    \caption{}\label{fig:prec-a}
   \end{subfigure}
   \centering
    \begin{subfigure}{0.45\textwidth}
        \includegraphics[width=0.9\textwidth,angle=270]{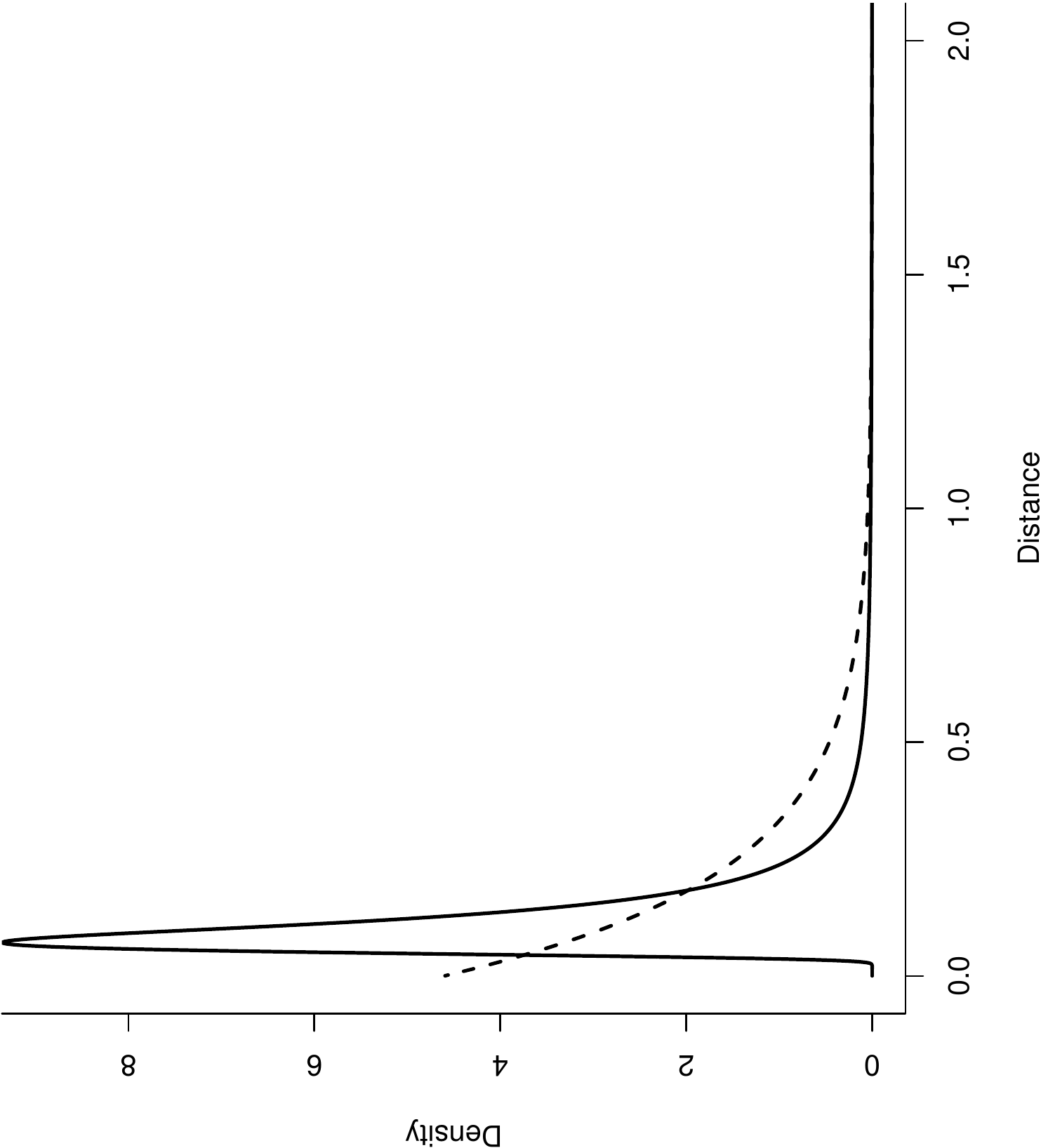}

 \caption{}\label{fig:prec-b}
    \end{subfigure}
        \caption{Panel~(a) displays the new prior (dashed) with parameters
        $(U=0.3/0.31, \alpha=0.01)$, and the $\Gamma(\text{shape}=1,
        \text{rate}=b)$ prior (solid). The value of $a$ is computed
        so that the marginal variances for the random effects are equal
        for the two priors, which leads to $b=0.0076$. Panel~(b) shows
        the same two priors on the distance scale demonstrating that
        the density for the Gamma prior is zero at distance zero.}
    \label{fig:prec}
\end{figure}

We end with a cautionary note about scaling issues for these
models and our third desideratum. If \mm{R} is full-rank, then it is usually scaled, or can be
scaled, so that $(\mm{R}^{-1})_{ii}=1$ for all $i$, hence $\tau$
represents the marginal precision. This leads to a simple
interpretation of $U$.  However, this is usually not the case if
$\mm{R}$ is singular like for spline and smoothing components; see
\cite{art521} for a discussion of this issue. Let the columns of
$\mm{V}$ represent the null-space of $\mm{R}$, so that $\mm{R}\mm{V} =
\mm{0}$. For smoothing spline models, the null-space is a
low-dimensional polynomial and \mm{R} defines the penalty for
deviating from the null space \citep[Sec.~3]{book80}. In order to
unify the interpretation of $U$, we can scale \mm{R} so that the
geometric mean (or some typical value) of the marginal variances of
$\mm{x} | \mm{V}^{T}\mm{x} = \mm{0}$ is one. In this way, $\tau$
represents the precision of the (marginal) deviation from the null
space.
%



\section{Some properties of \PCP s}
\label{sec3}

In this section, we investigate some basic properties of \PCP s for simple models and,
in general, we look at how specific properties of priors affect the
statistical results.  In particular, we will investigate when the
behaviour in the neighbourhood of the base model or the tail behaviour
is important to obtain sensible results.  For most
moderate-dimensional models, we find that the behaviour at the base
model is fundamentally important, while the tail behaviour is less
important. In contrast, in very high-dimensional settings, we find
that a heavier tail than that implied by the principle of constant
rate penalisation is required for optimal inference.

For practical reasons, in this section we restrict ourselves to a much 
smaller set of models than in the rest of the paper.  Sections \ref{sec:jeffries}--\ref{sec:hypothesis_testing} focus on \emph{direct observations of 
a single component model}, while Sections \ref{section:risk}--\ref{sec:sparsity} focus on estimating the mean of a normal distribution with known variance.  None of these models fall within the family of realistically complicated models that are the focus of this paper. Unfortunately, there is very little theory for the types of hierarchical models we are considering, so we are forced to consider these simpler models in order to gain intuition for the more interesting cases.

\subsection{Behaviour near the base model} \label{sec:jeffries}

To understand the \PCP\ construction better, we can study what happens
near $\xi=0$ when it is a regular point, using the connection between
KLD and the Fisher information metric. Let $I(\xi)$ be the Fisher
information at $\xi$. Using the well known asymptotic expansion
\begin{displaymath}
    \KLN{\pi(\mm{x} \mid \xi)}{\pi(\mm{x} \mid \xi = 0)} = \frac{1}{2}
    I(0)\xi^{2} + \text{higher order terms},
\end{displaymath}
a standard expansion reveals that our new prior behaves like
\begin{displaymath}
    \pi(\xi) = I(\xi)^{1/2} \exp\left(
      -\lambda\; m(\xi)
    \right) + \text{higher order terms}
\end{displaymath}
for $\lambda\xi$ close to zero.  Here, $m(\xi)$ is the distance
defined by the metric tensor $I(\xi)$, $m(\xi) = \int_0^{\xi}
\sqrt{I(s)} ds$, using tools from information geometry.  Close to the
base model, the \PCP\ is a tilted Jeffreys' prior for $\pi(\mm{x} |
\xi)$, where the amount of tilting is determined by the distance on
the Riemannian manifold to the base model scaled by the parameter
$\lambda$. The user-defined parameter $\lambda$ thus determines the
degree of informativeness in the prior.

\subsection{Large sample behaviour under the base model}

A good check when specifying a new class of priors is to consider the
asymptotic properties of the induced posterior. In particular, it is
useful to ensure that, for large sample sizes, we achieve frequentist
coverage.  While the Bernstein-von Mises theorem ensures that, for
regular models, asymptotic coverage is independent of (sensible) prior
choice, the situation may be different when the true parameter lies on
the boundary of the parameter space. In most examples in this paper,
the base model defines the boundary of the parameter space and prior
choice now plays an important role~\citep{bochkina2012bernstein}.

When the true parameter lies at the boundary of the parameter space,
there are two possible cases to be considered.  In the regular case,
where the derivative of the log-likelihood at this point is
asymptotically zero, \citet{bochkina2012bernstein} showed that the
large-sample behaviour depends entirely on the behaviour of the prior
near zero.  Furthermore, if the prior density is finite at the base
model, then the large sample behaviour is identical to that of the
maximum likelihood estimator \citep{self1987asymptotic}. Hence
Principle 1 ensures that \PCP s induce the correct asymptotic
behaviour.  Furthermore, the invariance of our construction implies
good asymptotic behaviour for any reparameterisation.

\subsection{\PCP s and Bayesian hypothesis testing} \label{sec:hypothesis_testing}

\PCP s are not built to be hypothesis testing priors and we do not
recommend their direct use as such.  We will show, however, that they
lead to consistent Bayes factors and suggest an invariant, weakly
informative decision theory-based approach to the testing problem.
With an eye towards invariance, in this section we will consider the
re-parameterisation $\zeta = d(\xi)$.

In order to show the effects of using \PCP s as hypothesis testing
priors, let us consider the large-sample behaviour of the precise test
$\zeta=0$ against $\zeta>0$.  We can use the results of
\citet{bochkina2012bernstein} to show the following in the regular
case.
\begin{theorem}
    \label{thm:BF}
    Under the conditions of \citet{bochkina2012bernstein}, the Bayes
    factor for the test $H_0: \zeta=0$ against $H_1: \zeta>0$, is
    consistent when the prior for $\zeta$ does not overfit. That is,
    $B_{01} \rightarrow \infty$ under $H_0$ and $B_{01} \rightarrow 0$
    under $H_1$.
\end{theorem}
\citet{johnson2010use} point out for regular models, that the rates at
which these Bayes factors go to their respective limits under $H_0$
and $H_1$ are not symmetric. This suggests that the finite sample
properties of these tests will be suboptimal. The asymmetry can be
partly alleviated using the moment and inverse moment prior
construction of \citet{johnson2010use}, which can be extended to this
parameter invariant formulation in a straightforward way
\citep[see][]{rousseaumoment}.  The key idea of non-local priors is to
modify the prior density so that it is approximately zero in the
neighbourhood of $H_0$.  This forces a separation between the null and
alternative hypotheses that helps balance the asymptotic rates.
Precise rates are given in \Appendix{appendix:BF}.
 
The construction of non-local priors highlights the usual dichotomy
between Bayesian testing and Bayesian predictive modelling: in the
large sample limit, priors that lead to well-behaved Bayes factors
have bad predictive properties and vice versa.  In a far-reaching
paper, \citet{bernardo2011integrated} suggested that this dichotomy is
the result of asking the question the wrong way.  Rather than using
Bayes factors as an ``objective'' alternative to a proper decision
analysis, \citet{bernardo2011integrated} suggests that reference
priors combined with a well-constructed invariant loss function allows
for predictive priors to be used in testing problems.  This also
suggests that \PCP s can be used in place of reference priors to
construct a consistent, coherent and invariant hypothesis testing
framework based on decision theory.

\subsection{Risk results for the normal means model}
\label{section:risk}

A natural question to ask when presented with a new approach for
constructing priors is \emph{are the resulting estimators any good?}.
In this section, we investigate this question for the analytically
tractable normal means model:
\begin{equation}
    \label{eqn:normal_means}
    y_i | x_i, \sigma\sim {\mathcal N}( x_i, 1), \qquad
    x_i  | \sigma\sim {\mathcal N}(0, \sigma^2), \qquad
    \sigma \sim \pi_d(\sigma), \qquad i=1, \ldots, p.
\end{equation}
This model is the simplest one considered in this paper and gives us 
an opportunity to investigate whether constant rate
penalisation, which was used to argue for an exponential prior on the
distance scale, makes sense in this context.  For the precision
parameter of a Gaussian random effect, the distance parameter is the
standard deviation, $d=\sigma$, which allows us to leverage our
understanding of this parameter and consider alternatives to this
principle.

For an estimator $\delta(\cdot)$, define the mean-square risk as
$R(\mm{x}_0, \delta) = \E\left(\|\mm{x_0}-\delta(\mm{y})\|^2 \right)$,
where the expectation is taken over data $\mm{y} \sim N(\mm{x_0},\mm{I})$.  The standard
estimator $\delta_0(\mm{y}) = \mm{y}$ is the best invariant estimator
and obtains constant minimax risk $R(\mm{x}_0,\delta_0) =
p$. Classical results of
\citet{james1961estimation,stein1981estimation} show that this
estimator can be improved upon.  We will consider the risk properties
of the Bayes' estimators, which in this case is the posterior mean.

By noting that $\E(x_i|\mm{y}, \sigma) = y_i(1 - \E(\kappa| \mm{y}))$ for
the shrinkage parameter $\kappa = (1+\sigma^{2})^{-1}$,
\citet{polson2012half} derived the general form of the mean-square
risk.  Using a half-Cauchy distribution on the standard deviation
$\sigma$, as advocated by \citet{art547}, the resulting
density for $\kappa$ has a horseshoe shape with infinite peaks at zero
and one.  The estimators that come from this horseshoe prior have good
frequentist properties as the shape of the density of $\kappa$ allows
the component to have any level of shrinkage.  In general, the density
for $\kappa$ is related to $\pi_d(\sigma)$ by
\begin{displaymath}
    \pi_\kappa(\kappa) = \pi_d\left(\sqrt{\kappa^{-1} - 1}\right)
    \frac{1}{2\sqrt{\kappa^3 - \kappa^4}}.
\end{displaymath}
Straightforward asymptotics shows how the limit behaviour of
$\pi_d(\sigma)$ transfers into properties of $\pi_\kappa(\kappa)$.
\begin{theorem}
    \label{thm:shrinkage}
    If $\pi_d(\sigma)$ has tails lighter than a Student-t distribution with
    $2$ degrees of freedom, then $\pi_\kappa(0)=0$.  If
    $\pi_d(d) \leq \mathcal{O}(d)$ as $d\rightarrow 0$, then $\pi_\kappa(1) = 0$.
\end{theorem}

\begin{figure}[h]
    \centering
    \begin{subfigure}{0.45\textwidth}
        \centering
        \includegraphics[width=0.95\textwidth]{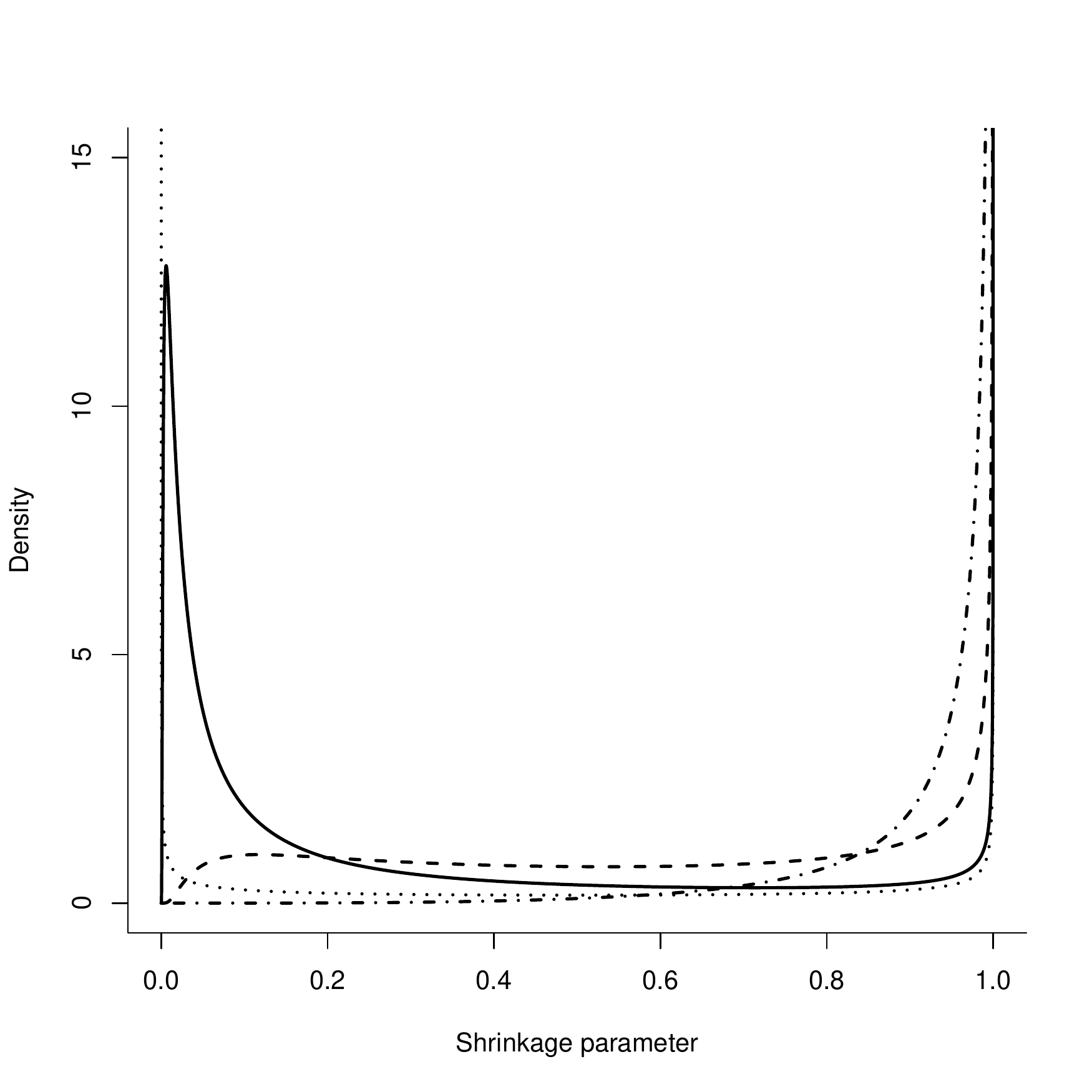}
        \caption{}
        \label{fig:riska}
    \end{subfigure}
    \quad
    \begin{subfigure}{0.45\textwidth}
        \centering \includegraphics[width=0.95\textwidth]{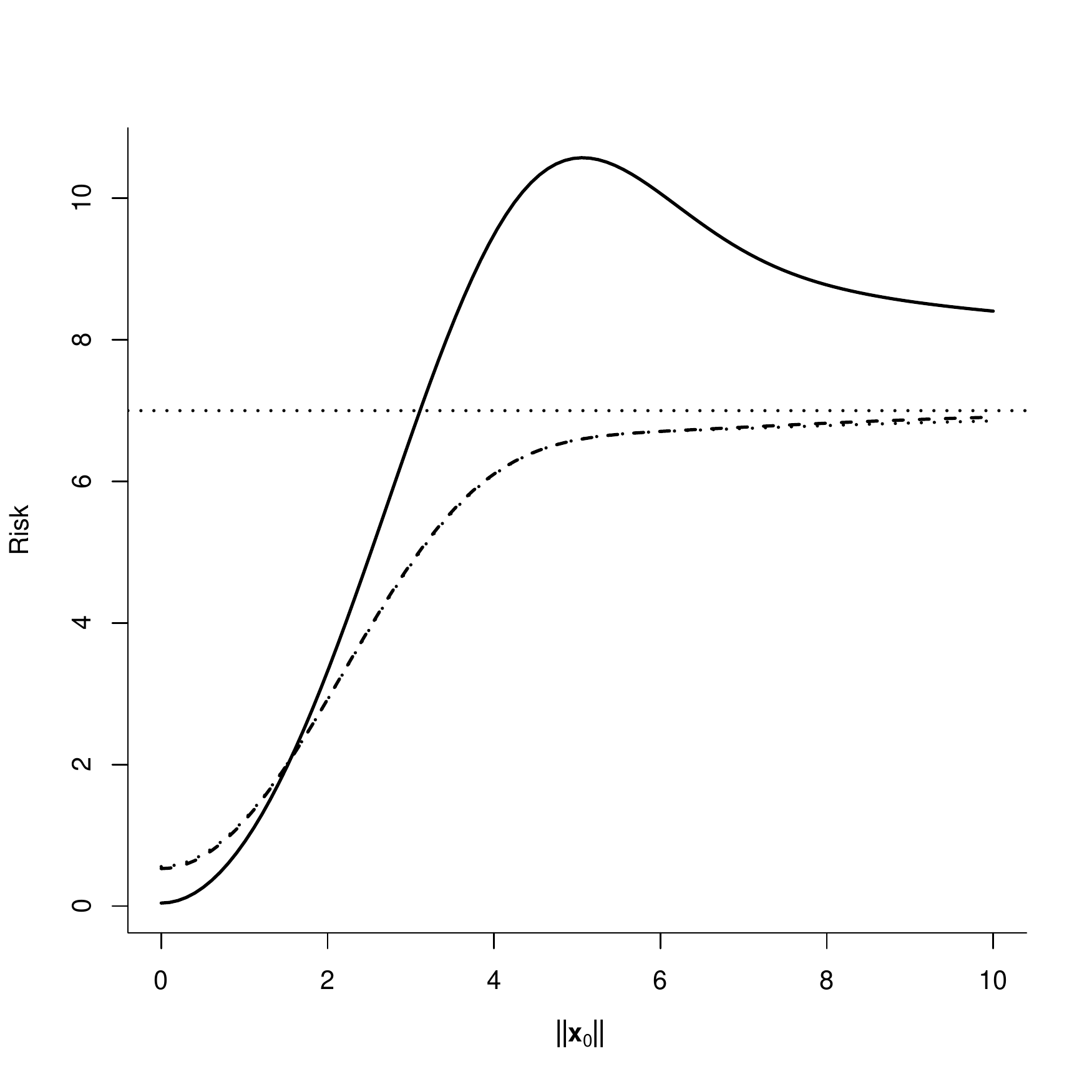}
        \caption{}
        \label{fig:riskb}
    \end{subfigure}
    \caption{Display~(a) shows the {implied prior on the shrinkage
            parameter $\kappa$ for several different priors on the
            distance scale. These priors are the half-Cauchy (dotted)
            and \PCP s with scaling parameter $\lambda=-\log(0.01)/U$
            for $U=1$ (dot-dashed), $U=5$ (dashed), and $U=20$
            (solid).}  Display~(b) shows the  mean squared risk of
            the Bayes' estimators for the normal means model with
            $p=7$ corresponding to different priors on the distance
            parameter, against $\|\mm{x}_0\|$.  The dotted line is the
            risk of the na\"{i}ve minimax estimator $\delta_0(\mm{x})
            = \mm{x}$.  The solid line corresponds to the \PCP\ with
            $U=1$. The dashed and dash-dot lines, which are
            essentially overlaid, correspond respectively to the \PCP\
            with $U=5$ and the half-Cauchy distribution.  }
\end{figure}

This result suggests that the \PCP\ will shrink strongly, putting
almost no prior mass near zero shrinkage, due to the relatively light tail of the
exponential. The scaling parameter $\lambda$ controls the decay of the
exponential, and the effect of $\lambda = -\log(\alpha)/U$, with
$\alpha=0.01$, on the implied priors on $\kappa$ is shown in
\Fig{fig:riska} for various choices of $U$.  For moderate $U$, the
\PCP\ still places a lot of prior mass near $\kappa=0$, in spite of
the density being zero at that point.  This suggests that the effect
of the light tail induced by the principle of constant penalisation
rate, is less than~\Thm{thm:shrinkage} might suggest. For comparison,
the horseshoe curve induced by the half-Cauchy prior is shown as the
dotted line in \Fig{fig:riska}. This demonstrates that \PCP s with
sensible scaling parameter place more mass at intermediate shrinkage
values than the half-Cauchy, which concentrates the probability mass
near $\kappa=0$ and $\kappa=1$.  The overall interpretation of \Fig{fig:riska} is that, 
for large enough $U$, the \PCP\ will lead to a slightly less efficient estimator 
than the half-Cauchy prior, while for small signals we expect them to behave similarly. 

\Fig{fig:riska} demonstrates also to which extent $U$ controls the
amount of information in the prior. The implied shrinkage prior for
$U=1$ (dot-dash line), corresponds to the weakly informative statement
that the effect is not larger than $3\sigma \approx 0.93$, has almost
no prior mass on $\kappa <0.5$.  This is consistent with the
information used to build the prior: if $\|\mm{x}_0\|<1$, the risk of
the trivial estimator $\delta(\mm{y}) = \mm{0}$ is significantly lower
than the standard estimator.

\Fig{fig:riskb} shows the risk using \PCP s with $U=1$ (solid line),
$U=5$ (dashed line), the half-Cauchy prior (dot-dashed line), as a
function of $\|\mm{x}_0\|$.  The mean-squared risk exceeds the minimax
rate for large $\|\mm{x}_0\|$ when $U=1$ which is consistent with the
prior/data mis-match inherent in badly mis-specifying $U=1$.  By
increasing $U$ to $5$, we obtain almost identical results to the
half-Cauchy prior, with a slight difference only for really large
$\|\mm{x}_0\|$.  Increasing $U$ decreases the difference.

The risk results obtained for the normal means model suggests that the
\PCP s give rise to estimators with good classical risk properties,
and that the heavy tail of the half-Cauchy is less important than the
finite prior density at the base model.  It also demonstrates that we
can put strong information into a \PCP, which we conjecture would be
useful for example Poisson and Binomial responses with link functions
like the log and logit, as we have strong structural prior knowledge
about the plausible range for the linear predictor in these cases \citep[Section 5]{polson2012half}.

\subsection{Sparsity priors}
\label{sec:sparsity}

When solving high-dimensional problems, it is often expedient to
assume that the underlying truth is sparse, meaning that only a small
number of the model components have a non-zero effect.  Good Bayesian
models that can recover sparse signals are difficult to build.
\citet{castillo2012needles} consider
spike-and-slab priors, that first select a subset of the components to
be non-zero and then place a continuous prior on these.  These priors
have been shown to have excellent theoretical properties, but their
practical implementation requires a difficult stochastic search
component.  A more pleasant
computational option builds a prior on the scaling parameter of the
individual model components.  In the common case where the component
has a normal distribution, the shrinkage properties of these priors
have received a lot of attention. Two examples of scale-mixtures of
normal distributions are the Horseshoe prior
\citep{carvalho2010horseshoe,van2014horseshoe} and the
Dirichlet-Laplace prior \citep{bhattacharya2014dirichlet} which both
have been shown to have comparable asymptotic behaviour to
spike-and-slab priors when attempting to infer the sparse mean of a
high dimensional normal distribution.  On the other hand,
\citet{castillo2014bayesian} showed that the Bayesian generalisation of
the LASSO \citep{park2008bayesian}, which can be represented as a
scale mixture of normals, gives rise to a posterior that contracts
much slower than the minimax rate.  This stands in contrast to the
frequentist situation, where the LASSO obtains almost optimal rates.

For concreteness, let us consider the problem
\begin{displaymath}
    \mm{y}_i \sim
    \pi(\mm{y} \mid \mm{\beta}), \qquad \mm{\beta} \sim {\mathcal N}(\mm{0},
    \mm{D}), \qquad {D}_{ii}^{-1} \stackrel{\text{iid}}{\sim}
    \pi(\tau),
\end{displaymath}
where $\pi(\mm{y}|\mm{\beta})$ is some data-generating distribution,
$\mm{\beta}$ is a $p$--dimensional vector of covariate weights,
$\pi(\tau)$ is the \PCP\ in \eqref{eq:prec} for the precisions
$\{D_{ii}^{-1}\}$ of the covariate weights. Let us assume that the 
observed data was generated from the
above model with true parameter $\mm{\beta}_0$ that has only $s_0$
non-zero entries. We will assume that $s_0 = {o}(p)$. Finally, in
order to ensure \emph{a priori} exchangeability, we set the scaling
parameter $\lambda$ in each \PCP\ to be the same.

This then begs the question: does an exponential prior on the standard
deviation, which is the \PCP\ in this section, make a good shrinkage
prior? In this section we will show that the answer is no. 
The problem with the basic \PCP\ for this problem is that the base
model has been incorrectly specified.  The base model that a $p$--dimensional
vector is sparse is not the same as the base model that each of the
$p$ components is independently zero and hence the prior encodes the wrong information.
  A more correct application of the principles in \Sec{sec:principles} would lead to a \PCP\ that
first selects the number of non-zero components and then puts i.i.d.\
\PCP s on each of the selected components.  If we measure complexity
by the number of non-zero components, the principle of constant rate
penalisation requires an exponential prior on the number of
components, which matches with the theory of
\citet{castillo2012needles}.  Hence, the failure of $p$ independent \PCP s to capture
sparsity is not unexpected.

To conclude this section, we show the reason for the failure of independent \PCP s to capture sparsity.
The problem is that the induced prior over $\mm{\beta}$ must have mass on values with a few large and many small components.  \Thm{thm:expon_decay} shows that the values of $\lambda$ 
that puts sufficient weight on approximately sparse models does not allow these models to have any large components.
Fortunately, the principled approach allows us to fix the problem by
simply replacing the principle of constant rate penalisation with
something more appropriate (and consistent with D8).  Specifically, in order for the prior to put appropriate mass around models with the true sparsity, the prior on the standard deviation needs to have a heavier tail than an exponential.

As $\pi(\tau)$ is an absolutely continuous distribution, the na\"ive
\PCP\ will never result in exactly sparse signals.  This leads us to
take up the framework of \citet{bhattacharya2014dirichlet}, who
consider the $\delta$--support of a vector
\begin{displaymath}
    \supp_\delta(\mm{\beta}) = \{i: |\beta_i|>\delta\},
\end{displaymath}
and define a vector $\mm{x}$ to be $\delta$--sparse if
$|\supp_\delta(\mm{\beta})| \ll p$.  Following
\citet{bhattacharya2014dirichlet}, we take $\delta =
\mathcal{O}(p^{-1})$. As $s_0 =o(p)$, this ensures that the non-zero
entries are small enough not to have a large effect on
$\norm{\mm{\beta}}$.

For fixed $\delta$, it follows that the $\delta$--sparsity of
$\mm{\beta}$ has a $\text{Binomial}(p, \alpha_p)$ distribution, where $\alpha_p =
\Prob(|\beta_i| > \delta_p)$.  If we had access to an oracle that told
us the true sparsity $s_0$, it would follow that a good choice of
$\lambda$ would ensure $\alpha_p = p^{-1}s_0$.
\begin{theorem}
    \label{thm:expon_decay}
    If the true sparsity $s_0 = o(p)$, then the oracle value of
    $\lambda$ that ensures that the $p^{-1}$--sparsity of $\mm{\beta}$
    is \emph{a priori} centred at the true sparsity grows like $
    \lambda \sim \mathcal{O}\left(\frac{p}{\log(p)}\right)$.
\end{theorem}
Theorem~\ref{thm:expon_decay} shows that $\lambda$ is required to
increase with $p$, which corresponds to a vanishing upper bound $U =
\mathcal{O}(p^{-1} \log(p))$.  Hence, it is impossible for the above
\PCP\ to have mass on signals that are simultaneously sparse and
moderately sized.

The failure of \PCP s to provide useful shrinkage priors is
essentially down to the tails specified by the principle of constant
rate penalisation.  This principle was designed to avoid having to
interpret a change of concavity on the distance scale for a general
parameter.  However, in this problem, the distance is the standard
deviation, which is a well-understood statistical quantity.  Hence, it
makes sense to put a prior on the distance with a heavier tail in this
case.  In particular, if we use a half-Cauchy prior in place of an
exponential, we recover the horseshoe prior, which has good shrinkage
properties. In this case \Thm{thm:general_sparse}, which is a
generalisation of \Thm{thm:expon_decay}, shows that the inverse
scaling parameter of the half-Cauchy must be at least
$\mathcal{O}(p/\log(p))$, which corresponds up to a log factor with
the optimal contraction results of \citet{van2014horseshoe}.  We note
that this is the only situation we have encountered in which the
exponential tails of \PCP s are problematic.



\section{The Student-t case}
\label{sec:t}

In this section we will study the Student-t case focusing solely on
the degrees of freedom (d.o.f.) parameter $\nu = 1/\xi$, keeping the
precision fixed. This is an important non-trivial case, since the
Student-t distribution is often used to robustify the Gaussian
distribution. Inference based on the Gaussian distribution is well
known to be vulnerable to model deviation. This can result in a
significant degradation of estimation performance
\citep{art538,art539,art540}, and using the Student-t distribution can
account for deviations caused by heavier tails, see for example
\citet{art545}, \citet{art541} and \citet{art565} for applications in
the econometric literature.

The base model for the Student-t distribution is the Gaussian, which
occurs when $\nu = \infty$. To maintain the interpretability of the
precision parameter of the distribution when $\nu < \infty$, we will
use a standardised version of the Student-t with unit precision for
all $\nu > 2$. This follows the advice from~\citet{art543}, promoting
that a parameter with a more orthogonal interpretation will ease the
(later) joint prior specification of the precision and the d.o.f.. Our
interpretation of the Occam's razor principle implies that the mode of
$\pi_d(d)$ must be at $d=0$, corresponding to the Gaussian
distribution. It turns out that any proper prior for $\nu$ with finite
expectation violates this principle and promotes overfitting, as
$\pi_d(0) = 0$.
\begin{theorem}
    \label{theorem:1}
    Let $\pi_{\nu}(\nu)$ be an absolutely continuous prior for $\nu>2$ where
    $\E(\nu)<\infty$, then $\pi_d(0) = 0$ and the prior overfits in the sense of  \Def{def:overfit}.
\end{theorem}
The proof is given in \Appendix{sec:proof}.  The intuition is that if
we want $\nu = \infty$ to be central in the prior, a finite
expectation will bound the tail behaviour so that we cannot have the
mode (or a non-zero density) at $d=0$.

Commonly used priors for $\nu$ include the exponential~\citep{art544}
or the uniform (on a finite interval)~distribution \citep{art545},
which, however, place zero density mass onto the base model causing
potentially overfitting according to Theorem~\ref{theorem:1}.  Notable
exceptions are the work of~\citet{art546} who computed (various forms
of) Jeffreys' priors in the case of linear regression models with
Student-t errors, \citet{art565} and \citet{art563} who use a proper
prior with no integer moments, and \citet{art559} who provide an
objective prior for discrete values of the d.o.f.; 

Consider the exponential prior for $\nu > 2$ with mean equal 
to $5$, $10$ and $20$. \Fig{fig:2}~(a) displays these priors converted to the 
distance scale $d = \sqrt{2\;\text{KLD}}$.
Similarly, \Fig{fig:2}~(b) displays the corresponding priors resulting from 
uniform priors on $\nu=2$ to $20$, $50$ and $100$.
\begin{figure}[tbp]
\begin{subfigure}{0.3\textwidth}
\includegraphics[width=0.9\textwidth,angle=270]{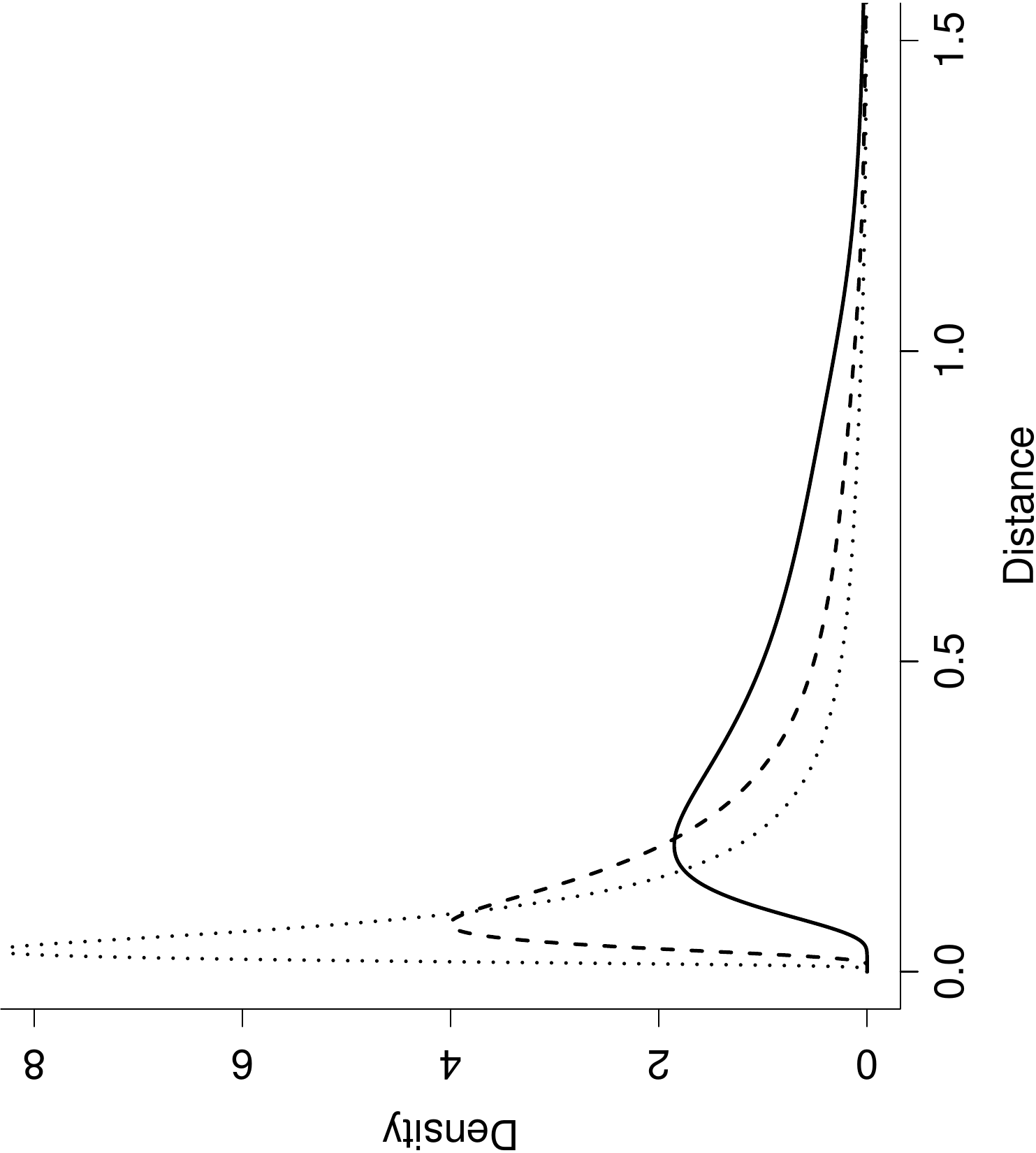}
\caption{}
        \label{fig:ta}
    \end{subfigure}
\begin{subfigure}{0.3\textwidth}
\includegraphics[width=0.9\textwidth,angle=270]{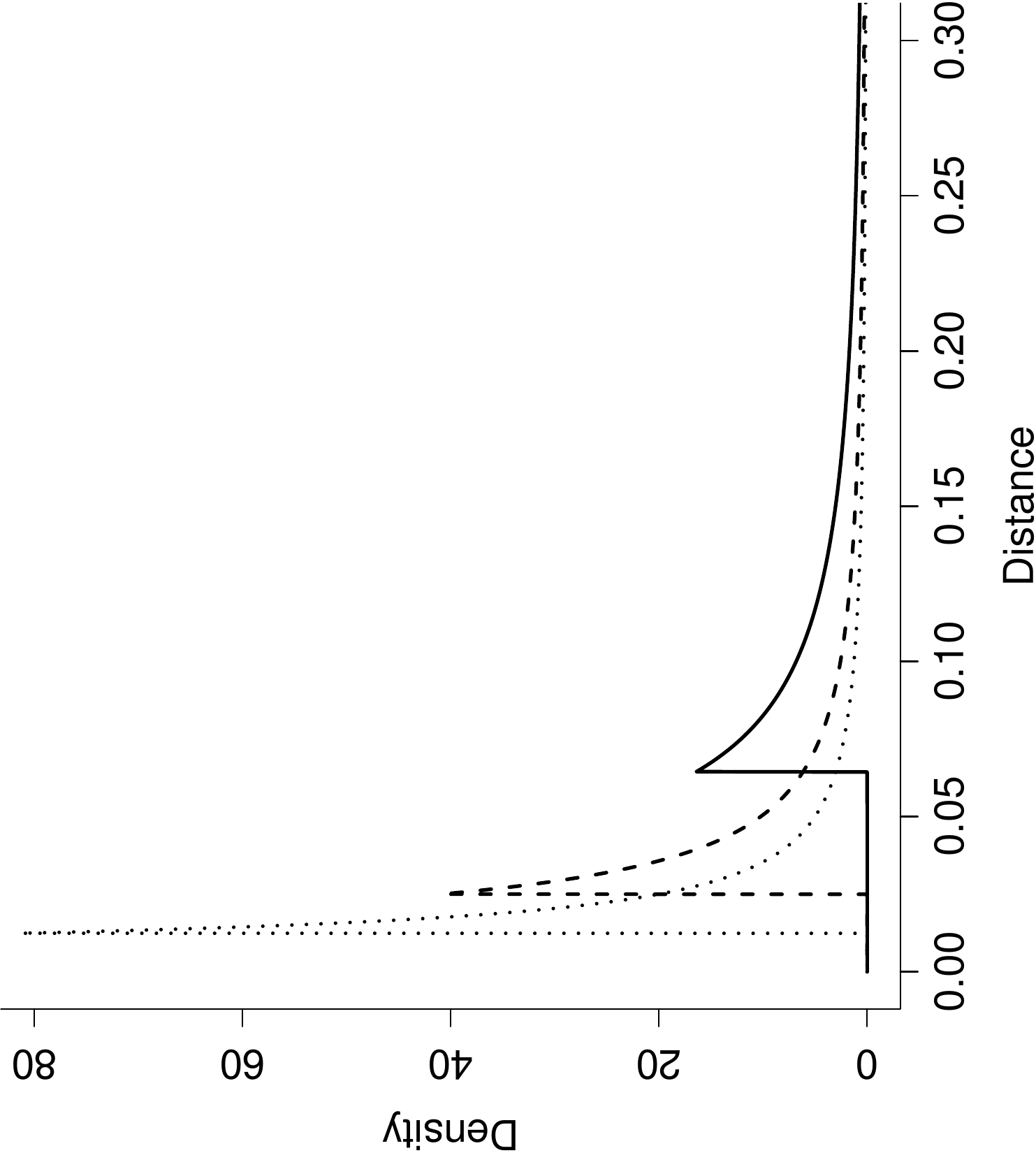}
\caption{}
        \label{fig:tb}
    \end{subfigure}
\begin{subfigure}{0.3\textwidth}
\includegraphics[width=0.9\textwidth,angle=270]{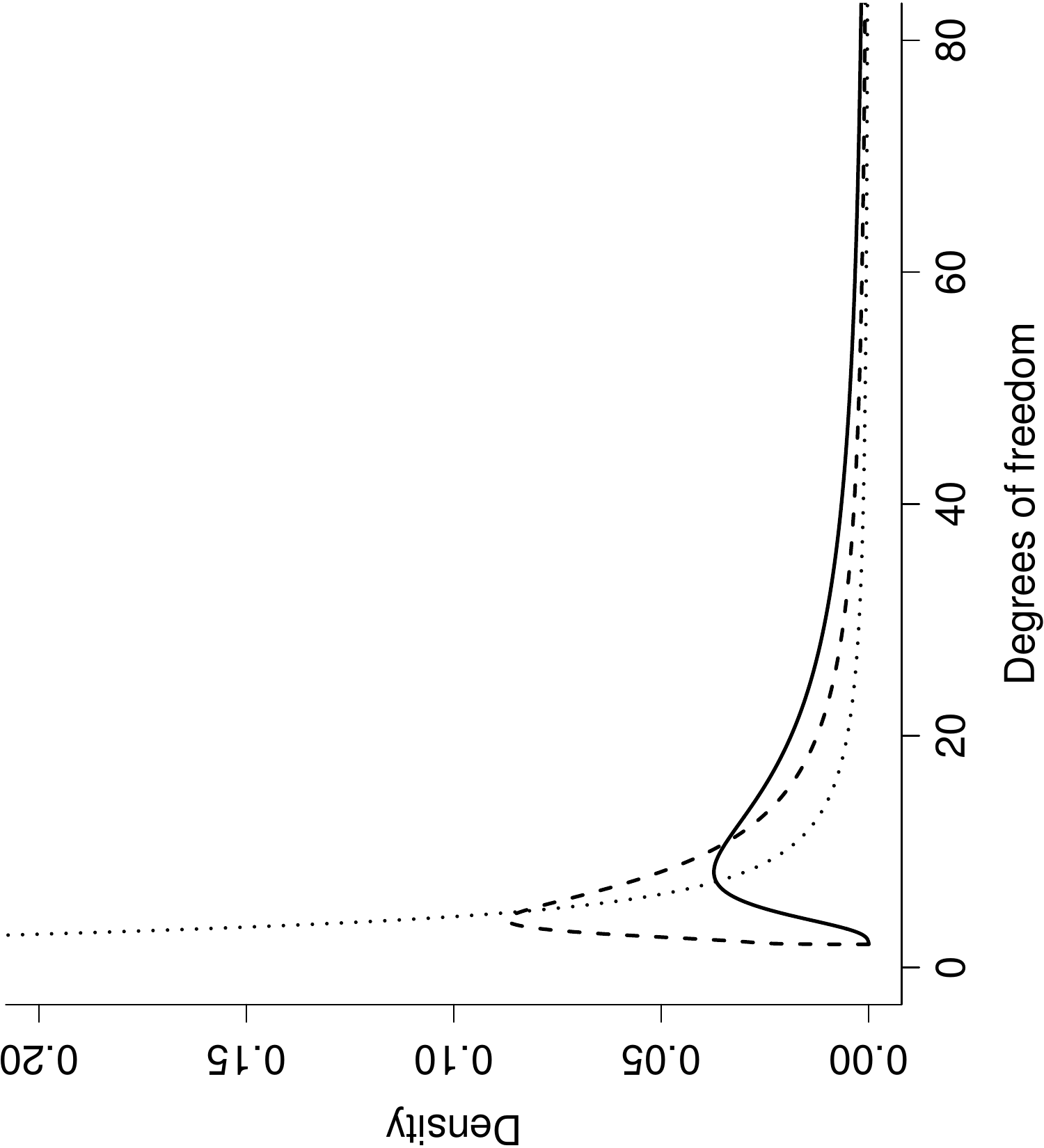}
\caption{}
        \label{fig:tc}
    \end{subfigure}
        
    \caption{Panel~(a) shows the exponential prior for $\nu>2$  with mean 
        equal to 5 (solid), 10 (dashed) and 20 (dotted) transformed to the
        distance scale.  Similarly, Panel~(b) shows
        the uniform prior for $\nu$ from 2, to $20$ (solid),
        $50$ (dashed) and $100$ (dotted) on the distance
        scale. Panel~(c) shows the PC priors
        for $\nu$ with $U=10$ and $\alpha=0.2$ (solid),
        $\alpha=0.5$ (dashed) and $\alpha=0.8$ (dotted) on the d.o.f scale.}
    \label{fig:2}
\end{figure}
As predicted by Theorem~\ref{theorem:1}, the density at $d \rightarrow
0$ is zero for all the six priors.  For the exponential priors in
panel~(a), the mode corresponds to $\nu = 7.0$, $17.0$ and $37.0$,
respectively.  This implies that Occam's razor does not apply, in the
sense that the exponential prior \emph{defines} that the posterior
will shrink towards the respective mode rather than towards the
Gaussian base model. The uniform prior behaves similarly as we
increase the upper limit, we put more mass to large d.o.fs and the
mode moves to the left. However, the finite support implies that the
density is zero up to the point defined by the upper limit. If the
\emph{true} distribution was Gaussian then we would overfit the data
using any of these priors.

The PC prior instead is defined to have the mode at $d=0$. To choose
the parameter $\lambda$ for the exponential distribution for $d$, a
notion of \emph{scale} is required from the user.  A simple choice is
to provide $(U, \alpha)$ such that $\Prob( \nu < U) =
\alpha$, giving $\lambda =
-\log(\alpha)/d(U)$. \Fig{fig:2}~(c) shows the corresponding
priors for $\nu$ setting $U=10$ and $\alpha=0.2, 0.5$ and
$0.8$.  Here, increasing $\alpha$ implies increasing the deviance from
the Gaussian base model.

To investigate the properties of the PC prior on $\nu$ and compare it
with the exponential prior on $\nu$, we performed a simulation
experiment using the model $y_i = \epsilon_i$, $i=1, \ldots, n$, where
$\epsilon$ is Student-t distributed with unknown d.o.f.\,and fixed
unit precision. Similar results are obtained for more
involved models \citep{tech117}.  We simulated data sets with $n=100,
1\,000, 10\,000$.  For the d.o.f.\,~we used $\nu = 5, 10, 20, 100$, to
study the effect of the priors under different levels of the kurtosis.
For each of the $12$ scenarios we simulated $1\,000$ different data
sets, for each of which we computed the posterior distribution of
$\nu$. Then, we formed the equal-weight mixture over all the $1\,000$
realisations to approximate the expected behaviour of the posterior
distribution over different realisations of the data. \Fig{fig:3}
shows the $0.025$, $0.5$ and $0.975$-quantiles of this mixture of
posterior distributions of $\nu$ when using the PC prior with
$U=10$ and $\alpha = 0.2, 0.3, 0.4, 0.5, 0.6, 0.7$ and $0.8$,
and the exponential prior with mean $5$, $10$, $20$ and $100$.  Each
row in~\Fig{fig:3} corresponds to a different d.o.f.\ while each
column corresponds to a different sample size~$n$.

\begin{figure}[tb]
    \centering \includegraphics[width=\textwidth]{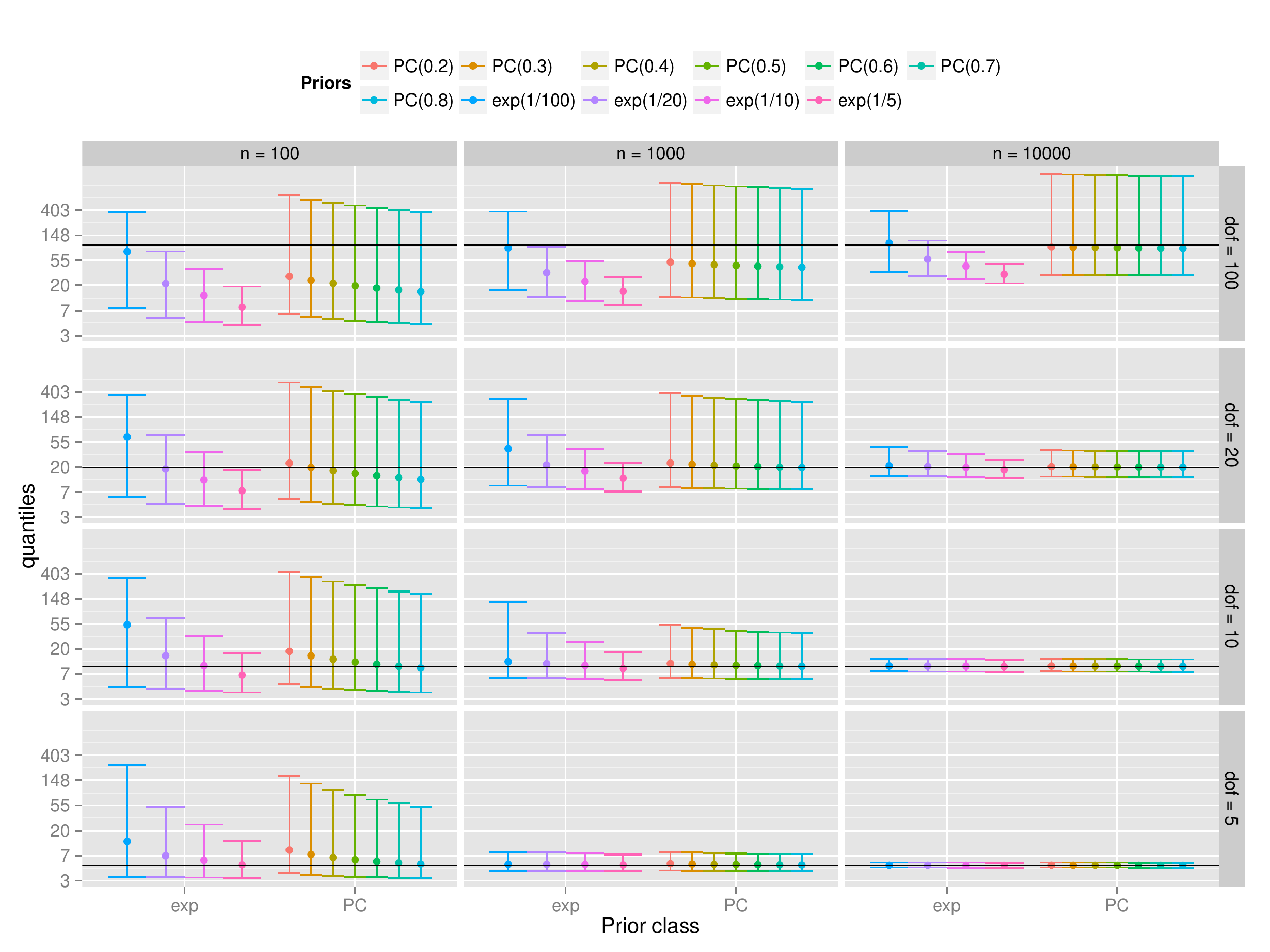}
    \caption{The $0.025$-, $0.5$- and $0.975$-quantile estimates
        obtained from an equal-weight mixture of posterior
        distributions of $\nu$ when fitting a Student-t errors model
        with different priors for $\nu$ over $1\,000$ datasets, for
        each of the $12$ scenarios with sample sizes $n = 100, 1\,000,
        10\,000$ and d.o.f.\ $\nu = 100, 20, 10, 5$.  The first four
        intervals in each scenario correspond to exponential priors
        with mean $100$, $20$, $10$, $5$, respectively.  The last
        seven intervals in each scenario correspond to the PC prior
        with $U = 10$ and $\alpha = 0.2, 0.3, 0.4, 0.5, 0.6,
        0.7$ and $0.8$.  }
    \label{fig:3}
\end{figure}

The first row in~\Fig{fig:3} displays the results with $\nu=100$ in
the simulation which is close to Gaussian observations. Using the PC
priors results in wide credible intervals in the presence of few data
points, but as more data are provided the model learns about the high
d.o.f..  Using an exponential prior for $\nu$, the posterior quantiles
obtained depend strongly on the mean of the prior.  This difference
seems to remain even with $n=1\,000$ and $n=10\,000$, indicating that
the prior still dominates the data. For all scenarios the intervals
obtained with the exponential prior for $\nu$ look similar, with the
exception of scenarios with low d.o.f.\ and high sample size, for
which the information in the data is strong enough to dominate this
highly informative prior.

If we study \Fig{fig:3} column-wise and top-down, we note that the
performance of the \PCP s are barely affected by the change in
$\alpha$ and they perform well for all sample sizes.  For the
exponential priors when $n=100$, we basically see no difference in
inference for $\nu$ comparing the near Gaussian scenario ($\nu=100$)
with the strongly heavy tailed one $(\nu=5)$.  The implication is that
the results will be much more influenced by the choice of the mean in
the exponential prior than by the d.o.f.\ in the data.  Similarly, the
exponential priors continue to be highly informative even for large
sample sizes. This informative behaviour can be seen in particular in
the first row ($\nu = \infty$), and is consistent with
\Def{def:overfit} and Theorem~\ref{theorem:1} .

We also inspected the coverage at a 95\% level for all priors and
simulation settings. The coverage probabilities for all \PCP s were
very similar and always at least $0.9$, whereby they tended to be a
bit too high compared to the nominal level. For the exponential priors
the results are ambiguous, either the coverage probabilities are
sensible while still being higher than the nominal level or they are
far too low, in several settings even zero.

This example sheds light on the consistency issue discussed by
\citet[Ch.~1]{book117}. A prior distribution represents prior belief,
learnt before data is observed, but it also fully specifies the
Bayesian learning model. As more data arrives it is expected that the
learning model goes in the right direction. If it does not, then the
learning model (prior) has not been set well, even though the prior
might be appropriate as representing prior beliefs. We claim that
priors that satisfy~\Thm{theorem:1} and therefore do not respect the
Occam's razor principle will invariably lead to bad learning models.
\Fig{fig:3} illustrates this point for the case of exponential priors.


\section{Disease mapping using the BYM model}
\label{sec:bym}

The application to disease mapping using the popular
BYM-model~\citep{art149} is particularly interesting since we are
required to reparameterise the model to make the sequence of base
models interpretable and parameters have a natural orthogonal
interpretation. Mapping disease incidence is a huge field within
public health and epidemiology, and good introductions to the field
exist~\citep{book115,book116,col21,col34}.

Disease incidences are usually rare and, to anonymise the data, they
are commonly available on an aggregated level, e.g.\ some
administrative regions like county or post-code area. The basic model
for the observed counts $y_i$ in area $i$ with $i=1,\ldots, n$, is the
Poisson distribution and the observations are assumed to be
conditionally independent Poisson variables with mean $E_i
\exp(\eta_i)$ where $\{E_i\}$ are the expected number of cases. These
are precomputed taking the number of people in the areas, their age
distribution and so on, into account. In the BYM-model we define the
log relative risk as $\eta_i = \mu + \mm{z}_i^{T}\mm{\beta} + u_i +
v_i$ where $\mu$ is the overall intercept, $\mm{\beta}$ measures the
effect of possibly region specific covariates $\mm{z}_i$, $\mm{v}$ is
a zero mean Gaussian with precision matrix $\tau_v \mm{I}$ and
represents an unstructured random effect. In contrast, $\mm{u}$ is a
spatial component saying that nearby regions are similar.  A first
order intrinsic Gaussian Markov random field
model~\citep[Ch.~3]{book80} was introduced by~\citet{art149} as a
model for $\mm{u}$. Let ${\mathcal G}$ be the conditional independence
graph of $\mm{u}$, where $\partial i$ denotes the set of neighbours to
node $i$ and let $n_{\partial i}$ be the corresponding number of
neighbours.  The conditional distribution of $u_i$ is
\begin{displaymath}
    u_i\mid \mm{u}_{-i}, \tau_u \;\sim\; {\mathcal N}\left(
      \frac{1}{n_{\partial
              i}}
      \sum_{j\in \partial i} u_j,
      1/\left(n_{\partial i} \tau_u\right) 
    \right)
\end{displaymath}
where $\tau_u$ is the precision parameter; see~\citet[Ch.~3]{book80}
for details. The full conditionals say that the mean for $u_i$ is the
mean of the neighbours, with a precision proportional to the number of
neighbours. This model is intrinsic in the sense that the precision
matrix has a non-empty null-space. The null-space here is the
$\mm{1}$-vector, hence the density is invariant to adding a constant
to $\mm{u}$. To prevent confounding with the intercept, we impose the
constraint that $\mm{1}^{T}\mm{u}=0$, assuming that the graph has only
one connected component.

To complete the model, we need the prior specification for the
intercept and the fixed-effects $\mm{\beta}$, additional to the prior
for the two precision parameters $\tau_u$ and $\tau_v$. It is common
to assign independent gamma priors for $\tau_u$ and $\tau_v$, often of
the form $\Gamma(1,b)$ where the rate parameters $b_u$ and $b_v$ are
set to small values. These choices are usually rather ad-hoc; see for
example~\citet{art368} for a discussion.

There are two main issues with the BYM model and the choice of
priors. The first, related to Desideratum D3, is that the spatial component is not \emph{scaled}
(see \Sec{sec:pc.prec}).
The marginal variance after imposing the
$\mm{1}^{T}\mm{u}=0$ constraint is not standardised, meaning that any
recommended prior (like those suggested by~\citet{art368}) cannot be
transferred from one graph to another, since the characteristic
marginal variance depends on the graph~\citep{art521}.  The second
issue, related to Desideratum D2, is that the structured component $\mm{u}$ cannot be seen
independently from the unstructured component $\mm{v}$. The
independent random effects \mm{v} are somehow partially included in
the spatial random effect \mm{u} for the case where no spatial
dependence is found.  Therefore, a proper model should account for
this feature to avoid identifiability issues. This means that the
priors for $\tau_u$ and $\tau_v$ should be (heavily) dependent, and
not independent as it is usually assumed.

To resolve these issues, and prepare the model-structure for our new
priors, we will reparameterise the model. Let $\mm{u}^{*}$ denote a
standardised spatial component where the characteristic marginal
variance is one. We then rewrite the log relative risk as
\begin{equation}\label{eq:bym-repar}
    \eta_i = \mu + \mm{z}_i^{T}\mm{\beta} +
    \frac{1}{\sqrt{\tau}}\left(\sqrt{1-\phi}\,
      v_i + \sqrt{\phi}\, u_i^{*}\right),
\end{equation}
where $0\le\phi\le1$ is a mixing parameter and the precision $\tau$
controls the marginal precision.  The marginal precision contribution
from $\mm{u}^{*}$ and $\mm{v}$ is $1/\tau$, whereas the fraction of
this variance explained by the spatial term $\mm{u}^{*}$ and the
random effects $\mm{v}$, are $\phi$ and $1-\phi$, respectively. Note
that the two hyperparameters $(\tau,\phi)$ are almost orthogonal (in
interpretation) and naturally allow for independent prior
specification. A similar reparameterisation has been proposed by
\citet{dean-etal-2001}.  However, they did not assume a scaled
structured spatial effect, which is essential for controlling the
influence of the corresponding hyperprior
\citep{art521}. \citet{leroux-etal-2000} proposed a slightly different
reparameterisation, which has been widely promoted as an alternative
formulation to the standard BYM model, see for example
\citet{lee-2011, ugarte-etal-2014}. The structured spatial effect is
however again not scaled and it is assumed that the precision matrix
of the new spatial model component is a weighted average of the
precisions of the structured and unstructured spatial components,
whereas here as well as in \citet{dean-etal-2001} this assumption is
posed on the variance scale.

With the reparameterisation proposed in \eqref{eq:bym-repar}, we can
apply our new approach to construct priors. First, we notice that the
type-2 Gumbel prior applies to the precision $\tau$, as the natural
base model is no effect from $\mm{u}^{*}$ and $\mm{v}$. For a fixed
marginal precision, the base model is no spatial dependency i.e.\
$\phi=0$.  An increased value of $\phi$ will blend in spatial
dependency keeping the marginal precision constant, hence more of the
variability will be explained by $\mm{u}^{*}$ and the ratio is
$\phi$. The prior for $\phi$ is derived in \Appendix{appendix:bym} and
depends on the graph ${\mathcal G}$. Our notion of \emph{scale} can be
used to set $(U, \alpha)$ so that $\Prob(\phi < U) = \alpha$ which
determines the degree of penalisation, see the \Appendix{appendix:bym}
for details.  This new reparameterisation should shrink towards both
no-effect and no-spatial effect and should therefore prevent
over-fitting of the data due to unfortunate \emph{ad hoc} prior
choices.

\subsection{Larynx data in Germany: spatial effect and the effect of
    an ecological covariate} \label{sec:bym_germany}

In this example, we will reanalyse the larynx cancer mortality for
men, registered in $544$ districts of Germany from $1986$ to
$1990$~\citep{art269}. The total number of deaths due to larynx cancer
was $7283$, which gives an average of $13.4$ per region. The
interesting part of this model is the semi-parametric estimation of
the covariate effect of lung cancer mortality rates in the same
period. This covariate acts as an ecological covariate~\citep{col35}
to account for smoking consumption, which is known to be the most
important risk factor of the larynx cancer.

As a smooth model for the ecological covariate $\mm{z}$,~\cite{art269}
used a second order random walk model
\begin{equation}\label{eq:rw2}%
    \pi(\mm{z} \mid \tau_z) \propto
    (\tau_z\tau_z^{*})^{(m-2)/2}\exp\left(
      -\frac{\tau_z\tau_z^{*}}{2} \sum_{i=3}^{m}\left(
        z_i-2z_{i-1}+z_{i-2}\right)^{2}
    \right)
\end{equation}
where $m$ is the length of $\mm{z}$. This spline model penalises the
estimated second order derivatives and its null space is spanned by
$\mm{1}$ and $(1, 2, \ldots, m)$. Similar to the spatial component in
the BYM model, this model component is not standardised and
$\tau_z^{*}$ ensures that the characteristic marginal variance is
one. The base model is here a straight line which reduces to a linear
effect of the ecological covariate, and the type-2 Gumbel distribution
is the resulting PC prior for $\tau_z$.

The log relative risk for this example is given by
\begin{equation}\label{eq:bym.spline}%
    \eta_i = \mu + f(z_i; \tau_z) + 
    \frac{1}{\sqrt{\tau}}\left(\sqrt{1-\phi}\,
      v_i + \sqrt{\phi}\, u_i^{*}\right)
\end{equation}
where $f(z_i; \tau_x)$ refers to the spline model at location $z_i$,
where the ecological covariate $\mm{z}$ have been converted into the
range $1, 2, \ldots, m$ for simplicity. See~\citet[Ch.~3]{book80} for
more details on this spline model, and~\citet{art435} for an extension
to irregular locations. For $\mm{u}^{*}$ and $\phi$ we used the same
parameters as in the previous example, and also for the precision in
the spline model we used $(U=0.2/0.31, \alpha=0.01)$. The results are
shown in~\Fig{fig:bym3}.

\begin{figure}[tbp]
\begin{subfigure}{0.3\textwidth}
    \includegraphics[width=0.9\textwidth,angle=270]{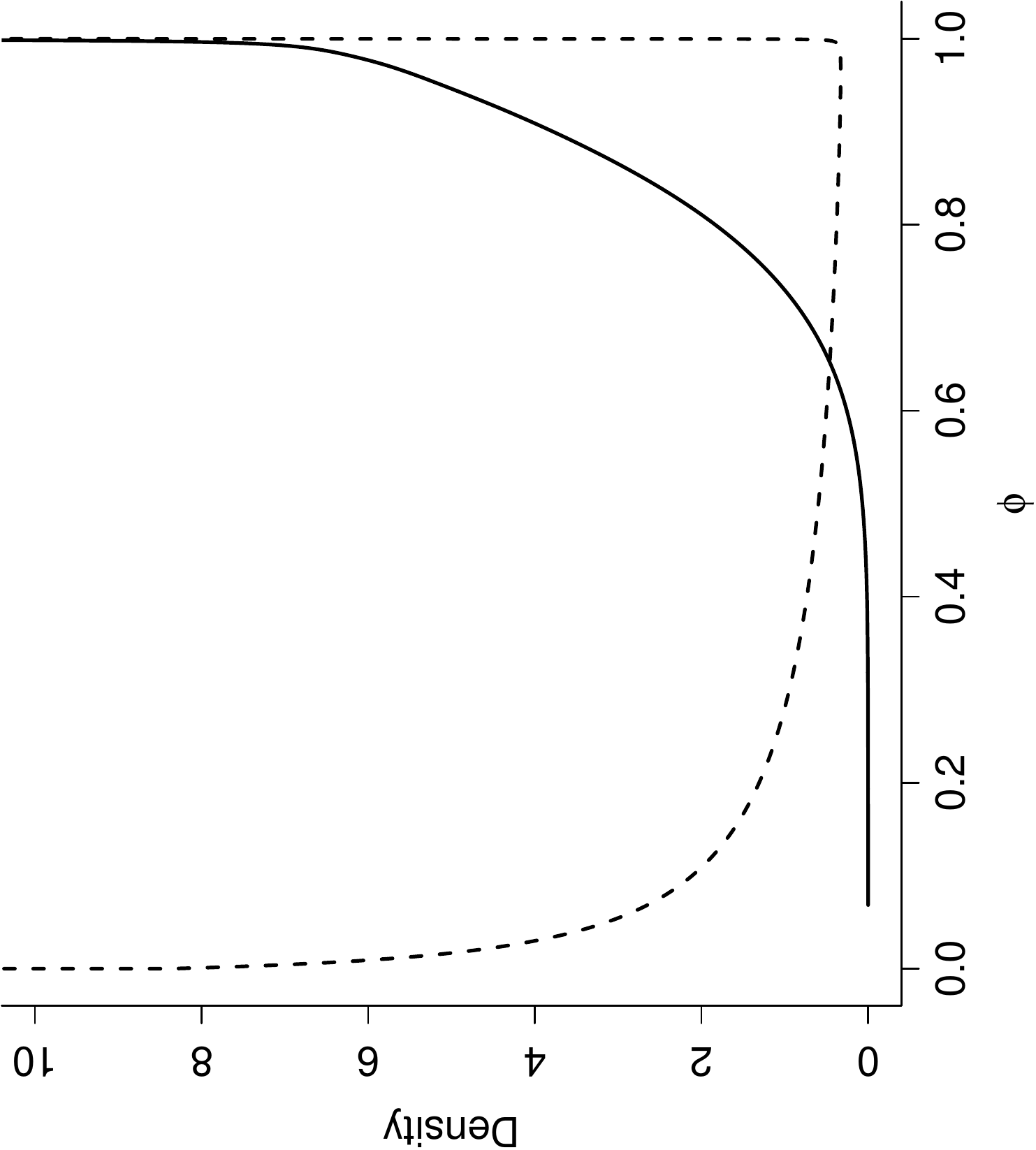}
\caption{}
        \label{fig:ga}
    \end{subfigure}
    \begin{subfigure}{0.3\textwidth}
        \includegraphics[width=0.9\textwidth,angle=270]{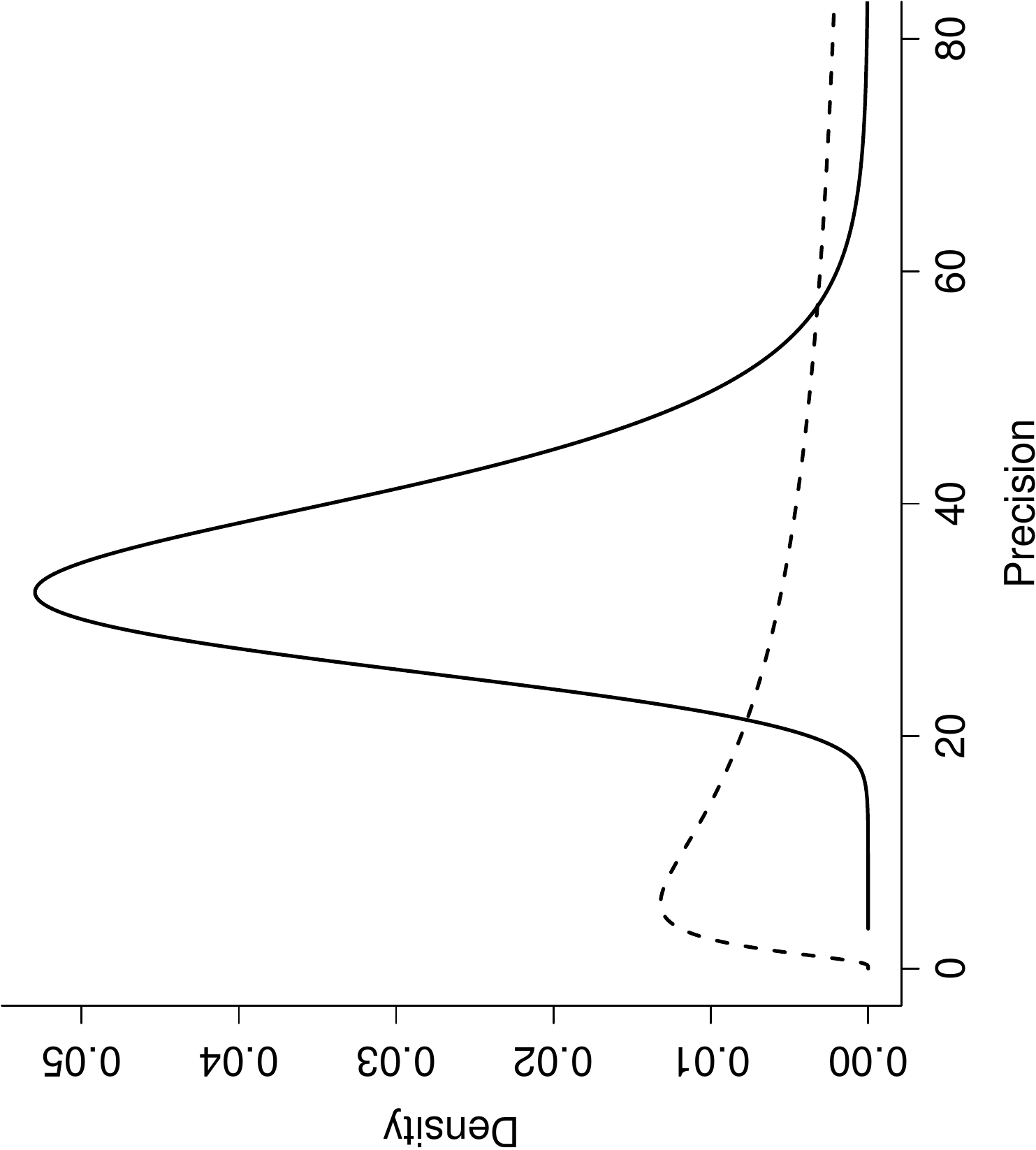}
\caption{}
        \label{fig:gb}
    \end{subfigure}
    \begin{subfigure}{0.3\textwidth}
        \includegraphics[width=0.9\textwidth,angle=270]{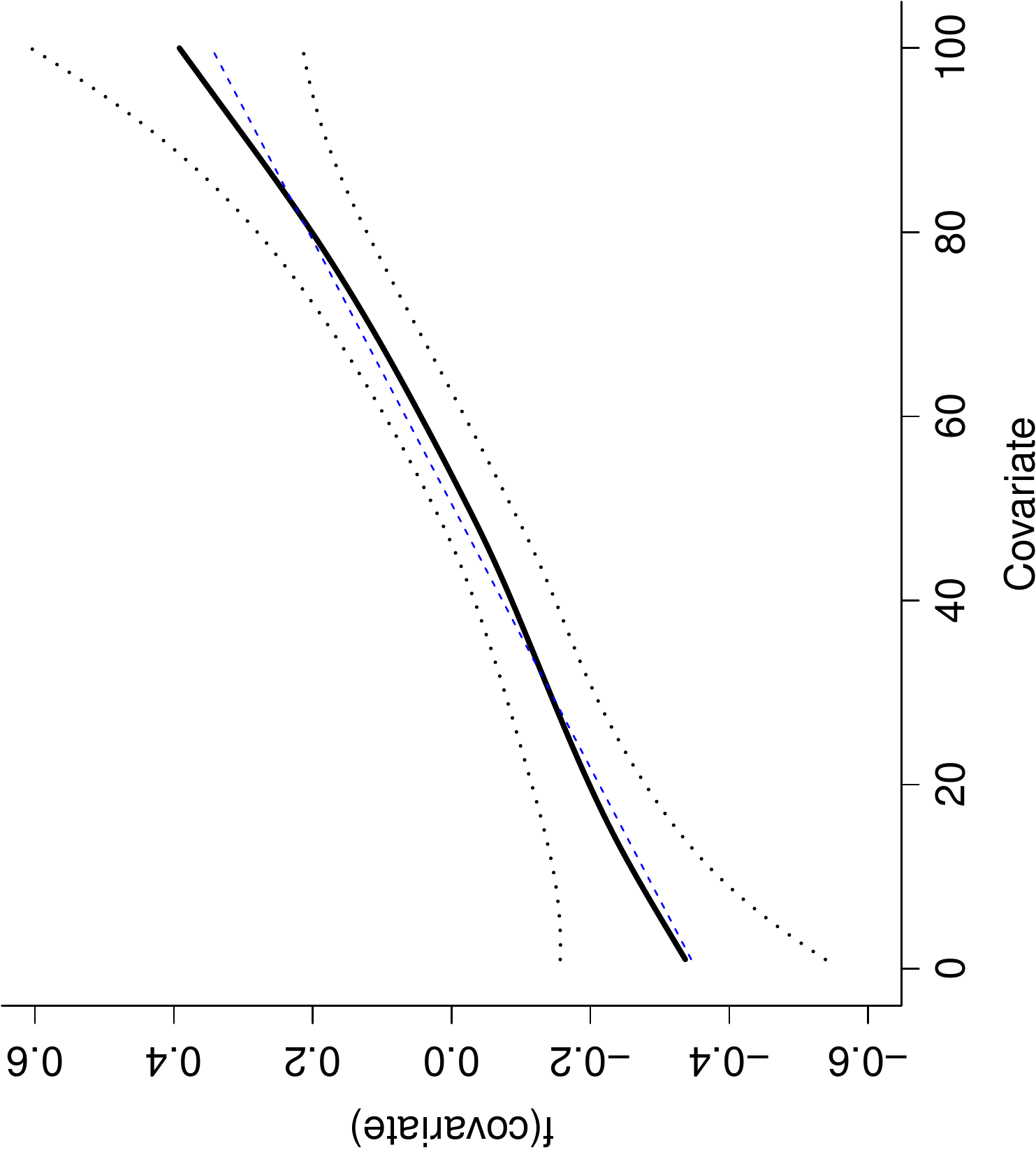}
    \caption{}
        \label{fig:gc}
    \end{subfigure}

    \caption{The results for the larynx data in Germany using PC
        priors and the reparameterised BYM model.  Panel~(a) shows the
        prior density for the mixing parameter $\phi$ (dashed) and the
        posterior density (solid). Panel~(b) shows the prior density
        (dashed) and the posterior density (solid) for the precision
        $\tau$. Panel~(c) shows the effect of the ecological covariate
        where the solid line is the mean, the dotted lines are the
        upper and lower $0.025$-quantiles, and the blue dashed line is
        the best linear fit to the mean.}
    \label{fig:bym3}
\end{figure}

\Fig{fig:bym3}~(a) shows that although the PC prior puts 2/3
probability on $\phi <1/2$, the model learns from the data resulting
in a posterior concentrated around 1. This implies that only the
spatial component contributes to the marginal variance. The posterior
for the precision $\tau$ (panel~(b)) is more concentrated than in
earlier examples due to the relatively high average counts. The effect
of the ecological covariate is shown in panel~(c) which displays the
mean (solid), the lower and upper $0.025$-quantiles (dotted) and the
best linear fit (blue dashed). The effect of the ecological covariate
seems to be shrunk towards the base model, i.e.~a straight line, and
is much more linear than the various estimates by \citet{art269}. We
suppose that the reason lies in over-fitting due to their choice of
priors.  The appropriateness of a linear effect of the ecological
covariate, is also verified using the approach in~\citet{art501}. Due
to the design of our new prior, the clear interpretation of the
parameters $(U, \alpha)$ and the good learning abilities demonstrated
by the new prior in previous examples, we do believe that the high
level of shrinking is a data driven effect and not a hidden effect of
the prior.


\section{Multivariate Probit Models}
\label{sec:mvprobit}

The examples considered thus far have been essentially univariate,
with higher dimensional parameter spaces dealt with using approximate
orthogonality.  In this section, we will demonstrate that the \PCP\
methodology naturally extends to multivariate parameters and
illustrate this by means of multivariate probit models.

Multivariate probit models have applications within sample surveys,
longitudinal studies, group randomised clinical trials, analysis of
consumer behaviour and panel data~\citep{art569}. They represent a
natural extension of univariate probit models, where the probability
for success at the $i$th subject is
\begin{eqnarray}\label{eq31}%
    \Prob(y_i = 1 \mid \mm{\beta}) = \Phi(\mm{x}_i^{T}\mm{\beta}),
    \qquad i=1, \ldots, n.
\end{eqnarray}
Here, $\Phi(\cdot)$ is the cumulative distribution function for the
standard Gaussian distribution, $\mm{x}_i$ a set of fixed covariates
with regression coefficients $\mm{\beta}$. The univariate probit model can be
reformulated into a latent variable formulation which both improves
the interpretation and eases computations. Let $z_i =
\mm{x}_i^{T}\mm{\beta} + \epsilon_i$, and define $y_i = 1$ if $z_i \ge
0$, and $y_i = 0$ if $z_i < 0$. When $\{\epsilon_i\}$ is standard
multivariate Gaussian over all the $n$ subjects, we obtain~\eref{eq31}
after marginalising out $\epsilon_i$.  In the multivariate extension
we have $m$ measurements of the $i$th subject, $\{y_{ij}: j=1,\ldots,
m\}$.  The latent vector for the $i$th subject is (with a slight
confusion in the notation) $\mm{z}_i = \mm{X}_i^{T}\mm{\beta} +
\mm{\epsilon_i}$ where $\mm{\epsilon_i} \sim {\mathcal N}_m(\mm{0},
\mm{R})$, and define $y_{ij}=1$ if $z_{ij} \ge 0$, and $y_{ij}=0$ if
$z_{ij} < 0$. The dependence within each subject, is encoded through
the matrix $\mm{R}$, which, in order to ensure identifiability, is
restricted to be a correlation matrix.

A Bayesian analysis of a multivariate probit model requires a prior
for the correlation matrix $\mm{R}$. For the saturated model for
\mm{R}, \citet{barnard2000modeling} demonstrate the joint uniform
prior $\pi(\mm{R})\propto 1$ which gives highly informative marginals
centred at zero; see \citet{art569}
for applications of this prior within multivariate probit models. The
joint Jeffreys' prior for \mm{R} was used by~\citet{art571}, which
(unfortunately) place most prior mass close to $\pm 1$ in high
dimension. \citet{art570} suggest using a multivariate Gaussian prior
for $\mm{R}$ restricted to the subset where $\mm{R}$ is positive
definite. Neither of these prior suggestions for \mm{R} are particular
convincing.

\subsection{Extending the univariate PC prior construction} \label{sec:multivar_PC}

The principles underlying the PC prior outlined in
\Sec{sec:principles} can be extended to the multivariate setting
$\mm{\xi} \in \mathcal{M}$ with base model $\mm{\xi} = \mm{0}\in
\mathcal{M}$.  This multivariate extension has all the features of the
univariate case.  As many interesting multivariate parameters spaces
are not $\mathbb{R}^n$, we will let $\mathcal{M}$ be a subset of a
smooth $n$-dimensional manifold.  For example, when modelling
covariance matrices $\mathcal{M}$ will be the manifold of symmetric
positive definite matrices, while the set of correlation matrices is a
convex subset of that space.  A nice introduction to models on
manifolds can be found in \citet{byrne2013geodesic}, where the problem
of constructing useful Markov chain Monte Carlo schemes is also
considered.

Assume that $d(\mm{\xi})$ has a non-vanishing Jacobian.  For each $r
\geq 0$, the level sets $\mm{\theta} \in S_r = \left\{ \mm{\xi} \in
  \mathcal{M} : d(\mm{\xi}) = r\right\}$ are a system of disjoint
embedded submanifolds of $\mathcal{M}$, which we will assume to be compact \citep[Chapter
8]{lee2003smooth}.  In the parlance of differential geometry, the
submanifolds $S_r$ are the leaves of a foliation and the decomposition
$\mathcal{M} = \mathbb{R}_+ \times \left(\sqcup_{r\geq 0} S_r\right)$
gives rise to a natural coordinate system on $\mathcal{M}$.  Hence the
natural lifting of the \PCP\ concept onto $\mathcal{M}$ is the prior
that is exponentially distributed in $d(\mm{\xi})$ and uniformly
distributed on the leaves $S_{d(\mm{\xi})}$.

In some sense, this above definition is enough to be useful.  A simple
MCMC or optimisation scheme would proceed in a ``coordinate ascent''
manner, moving first in the distance direction and then along the leaf
$S_r$.  More efficient schemes, however, may be derived from a more
explicit density.  To this end, we can locally find a mapping
$\varphi(\cdot)$ such that$
    \left(d(\mm{\xi}), \varphi(\mm{\xi})\right) = \mm{g}(\mm{\xi})$.
With this mapping, we get a local representation for the multivariate
\PCP\ as
\begin{equation}
    \label{eqn:multi_pcp}
    \pi(\mm{\xi})
    =\frac{\lambda}{\left|S_{d(\mm{\xi})}\right|}  \exp\left({-\lambda
          d(\mm{\xi})}\right) \left| \det(\mm{J}(\mm{\xi}) )\right |,
\end{equation}
where $J_{ij} = \frac{\partial g_i}{\partial \xi_j}$ is the Jacobian
of $\mm{g}$. While the definition of multivariate \PCP s is
significantly more involved than in the univariate case, it is still
useful.  In general, computational geometry can be used to evaluate
\eqref{eqn:multi_pcp} approximately in low dimension.  In the case
where the level sets are simplexes or spheres, exact expressions for
the \PCP\ can be found. These situations occur when $d(\mm{\xi})$ can
be expressed as
\begin{equation}\label{eq:h.simplex}%
    d(\mm{\xi}) = h\left(\mm{b}^{T}\mm{\xi}\right),
    \quad \mm{b} > \mm{0},
    \quad \mm{\xi} \in \mathbb{R}^{n}_+
\end{equation}
or
\begin{equation}\label{eq:h.sphere}%
    d(\mm{\xi}) = h\left(\frac{1}{2}\mm{\xi}^{T}\mm{H}\mm{\xi}\right),
    \quad \mm{H} > 0, \quad \mm{\xi}\in \mathbb{R}^{n},
\end{equation}
for some function $h(\cdot)$ satisfying $h(0)=0$, typically
$h(a)=\sqrt{2a}$.  The linear case will prove useful for deriving the
\PCP\ for general correlation matrices in \Sec{sec:PC_corr}.
  The quadratic case will be fundamental to derive approximate multivariate
\PCP s for hierarchical models, see \Sec{sec:ext}. 

It is trivial to
simulate from the \PCP\ when \eref{eq:h.simplex} or \eref{eq:h.sphere}
holds.  First we sample $d$ from the exponential distribution with
rate $\lambda$. In the linear case~\eref{eq:h.simplex}, sample
$\mm{s}$ uniformly on an $(n-1)-$simplex (by sampling $z_1, \ldots,
z_n$ independent from $\text{Exp}(1)$ and set $\mm{s} =
\mm{z}/\mm{1}^{T}\mm{z})$ and compute $\mm{\xi} =
h^{-1}(d)\mm{s}/\mm{b}$.  In the quadratic case~\eref{eq:h.sphere},
sample $\mm{s}$ uniformly on an unit sphere (by sampling independent
standard Gaussians $z_1, \ldots, z_n$ and then set $\mm{s} =
\mm{z}/\sqrt{\mm{z}^{T}\mm{z}}$) and compute $\mm{\xi} =
\sqrt{2h^{-1}(d)}\mm{H}^{-1/2}\mm{s}$. Using these direct simulation
algorithms, it is a simple change of variables exercise to derive the
densities for the \PCP s. In the linear case with $\mm{b}=\mm{1}$, the
\PCP\ is
\begin{equation}\label{eq:simplex}%
    \pi(\mm{\xi}) =
    \lambda \exp\left(-\lambda d(\mm{\xi})\right)
    \times
    \frac{(n-1)!}{r(\mm{\xi})^{n-1}}
    \times \left|h'(r(\mm{\xi}))\right|, \qquad
    r(\mm{\xi}) = h^{-1}(d(\mm{\xi})),
\end{equation}
in the quadratic case with $\mm{H}=\mm{I}$, the
\PCP\ is
\begin{equation}\label{eq:sphere}%
    \pi(\mm{\xi}) =
    \lambda \exp\left(-\lambda d(\mm{\xi})\right)
    \frac{\Gamma\left(\frac{n}{2}+1\right)}{
        n\pi^{\frac{n}{2}}\; r(\mm{\xi})^{n-2}}
    \left|h'\left(\frac{1}{2}r(\mm{\xi})^{2}\right)\right|, \,
        r(\mm{\xi}) = \sqrt{2h^{-1}(d(\mm{\xi}))}.
\end{equation}
The results for the general case, $\mm{b}>\mm{0}$ and $\mm{H}>0$,
follows directly after a linear transformation of $\mm{\xi}$.

\subsection{A prior for general correlation matrices} \label{sec:PC_corr}

Consider the model component $\mm{x} \sim {\mathcal N}(\mm{0},
\mm{R})$ where $\mm{R}$ is a $q \times q$ correlation matrix. The
distance function to the base model is given by $d(\mm{R}) = \sqrt{-
    \log(\det(\mm{R}))}$.  This distance function can be greatly
simplified by considering a different parameterisation of the set of
correlation matrices.  \citet{rapisarda2007parameterizing} show that
every correlation matrix can be written as $\mm{R} = \mm{BB}^T$, where
$\mm{B}$ is a lower triangular matrix with first row given by a $1$ on
the diagonal (first position) and zeros in every other position and,
for rows $i \geq 2$, entries are given by
\begin{displaymath}
    B_{ij} = \begin{cases}
        \cos(\theta_{ij}), \qquad &j=1; \\
        \cos(\theta_{ij}) \prod_{k=1}^{j-1} \sin(\theta_{ik}), 
        \qquad& 2 \leq j \leq i-1; \\
        \prod_{k=1}^{i-1} \sin(\theta_{ik}), \qquad &j=i; \\
        0, \qquad &i+1 \leq j \leq q,
    \end{cases}
\end{displaymath}
where $\theta_{ij} \in [0,\pi]$.  The advantage of this
parameterisation is that the distance function is now given by
$d(\mm{R}) = \sqrt{2\sum_{i=2}^q \sum_{j=1}^{i-1} \gamma_{ij}}$, where
$\gamma_{ij} = -\log(\sin (\theta_{ij})) \in [0,\infty)$ are the $p =
q(q-1)/2$ parameters. Using the $\mm{\gamma}$-parameterisation, we are
in the linear case \eref{eq:h.simplex} and the \PCP\ is given by
\eref{eq:simplex} with $\mm{\xi} = \mm{\gamma}$, $h(a) = \sqrt{2a}$
and $n = p$. The \PCP\ for $\mm{\theta}$ follows directly after a
change of variables exercise, and is simplified by noting that the two
branches of the mapping $\theta_{ij}=\theta_{ij}(\gamma_{ij})$ has
equal Jacobian.  

The scaling parameter $\lambda$ controls the degree of compression of
the parallelotope with vertices given by the column vectors of
$\mm{R}$.  For large values of $\lambda$, most of the mass will be
concentrated near parallelotopes with unit volume, while for small
$\lambda$, the volume could be significantly less than one.  This
parameter may be difficult to visualise in practice, and we suggest
calibrating the prior by drawing samples from the model component and
selecting a value of $\lambda$ for which this component behaves
appropriately.   Figure \ref{fig:corr}(a) shows the
\PCP\ marginal for one of these correlations for an exchangeable \PCP\
on a $3\times 3$ correlation matrix, using $\lambda = 10, 5$ and
$2$. Decreasing $\lambda$ makes the marginal less tightened up around
zero.

There are two complications in interpreting the \PCP\ for
$\mm{\theta}$. The first, as mentioned in the previous section, is
that the corresponding prior is not exchangeable due to the dependence
on the Cholesky decomposition on the ordering of the random
effects. This can be rectified by summing over all orderings, however
we have observed that the pairwise distributions are very similar even
without doing this. While we do not necessarily recommend summing out
the permutations in practice, for the figures in this section, we have
computed the exchangeable \PCP .

\begin{figure}[tbp]
    \centering
    \begin{subfigure}{0.45\textwidth}
        \centering
        \rotatebox{270}{\includegraphics[width=0.9\textwidth]{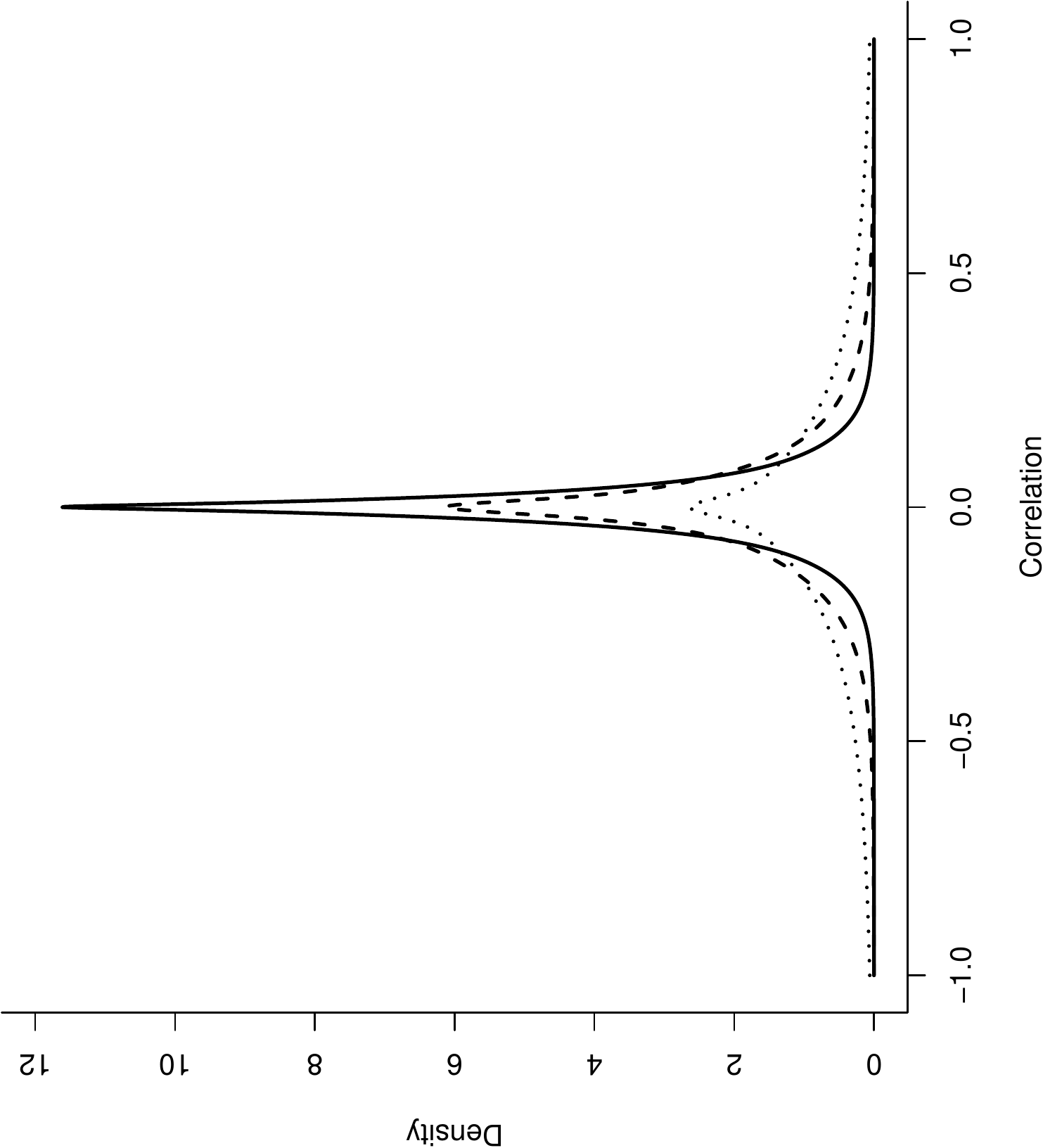}}
        \caption{}
    \end{subfigure}
    \quad
    \begin{subfigure}{0.45\textwidth}
        \centering
        \rotatebox{270}{\includegraphics[width=0.9\textwidth]{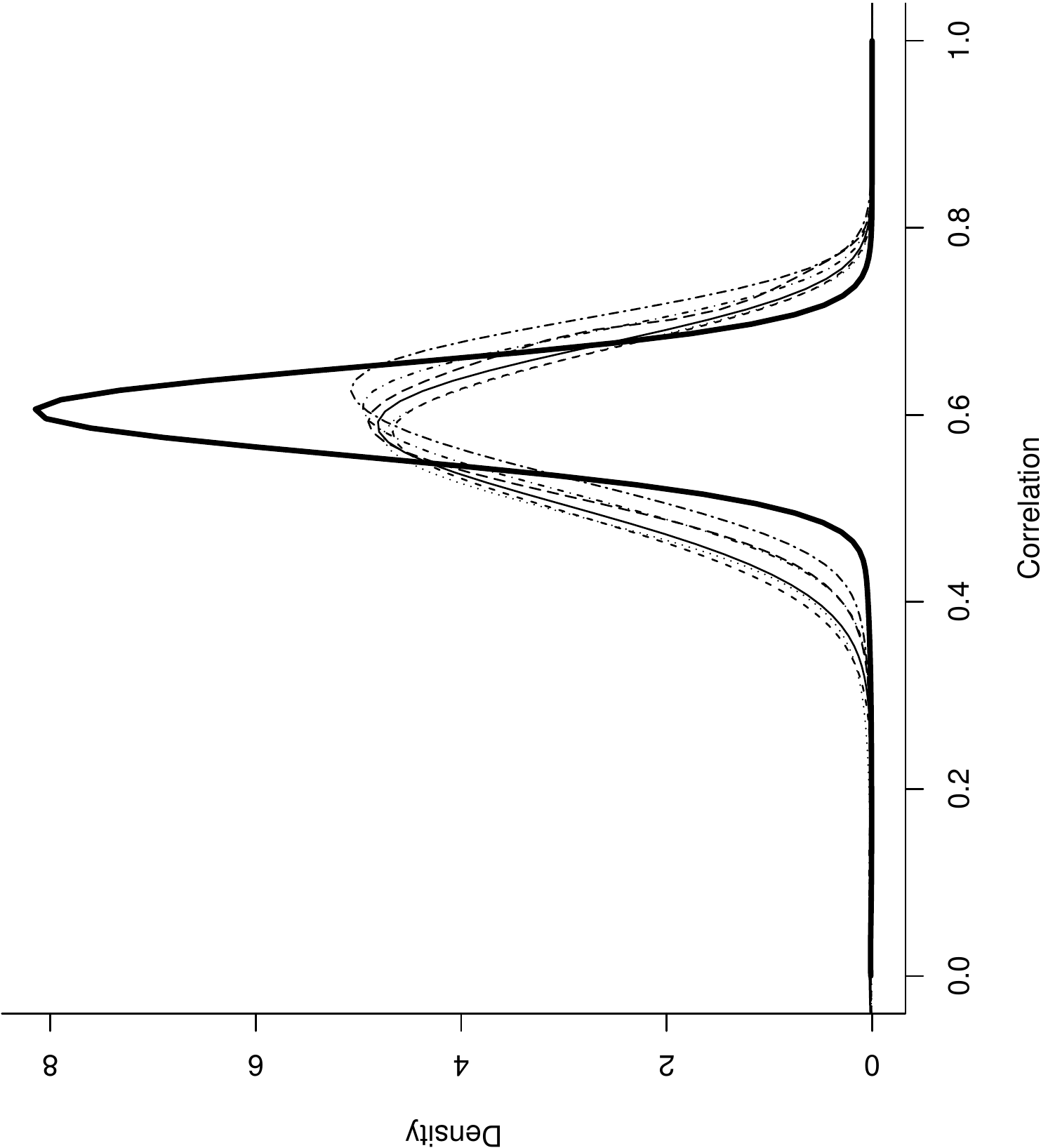}}
        \caption{}
    \end{subfigure}
    \caption{Panel~(a) shows the symmetric marginal prior density for
        the correlation computed from the \PCP\ for a $3\times 3$
        correlation matrix with $\lambda = 10$ (solid), $5$ (dashed)
        and $2$ (dotted).  Panel~(b) shows the posterior correlations
        for the Six Cities study with $\lambda=0.1$. The solid tick
        line is for the exchangeable model. The marginal densities are
        approximately identical.    \label{fig:corr}}
\end{figure}

\subsection{The Six Cities study and exchangeable random effects}

We will now illustrate the use of \PCP s for \mm{R} and reanalyse a
subset of the data from the Six Cities study discussed by
\citet[Sec~5.2]{art570} using the data as tabulated in their
Table~3. The Six Cities study investigates the health effects of air
pollution; refer to \citet{art570} for background. The response is the
wheezing status from $n=537$ children in Ohio at ages $7$, $8$, $9$
and $10$ years, and the aim is to study the effect of the smoking
habit of the mother to the response. The model is
\begin{displaymath}
    \Prob(y_{ij} = 1 \mid \mm{\beta}, \mm{R}) = \Phi\left(
      \beta_0 + \beta_1 x_{i1} + \beta_2 x_{i2} + \beta_3 x_{i3}
    \right), \qquad j = 1, \ldots, m=4,
\end{displaymath}
where covariates are $\mm{x}_1$ representing age (centred at $9$),
$\mm{x}_2$ for smoking ($1=$yes, $0=$no), and their interaction
$\mm{x}_3$, respectively. \citet{art570} used two models for $\mm{R}$,
the saturated or general case with $m(m-1)/2$ parameters, and the
exchangeable case where all off-diagonal terms in $\mm{R}$ are the
same. Our analysis is made more comparable to \citet{art570} by
adapting their ${\mathcal N}_4(\mm{0}, 10^{-1}\mm{I})$ prior for
$\mm{\beta}$.

We chose the decay-rate $\lambda$ by sampling from the \PCP\ for
various values of $\lambda$. We then estimated $\Prob(|\rho_{ij}| >
1/2)$ and found the values of $\lambda$ where this probability was
approximately $0.8$ and $0.2$. These two (rather extreme) choices gave
$\lambda = 0.1$ and $1.0$. We then ran two long MCMC chains to obtain
posterior samples after disregarding the burn-in phase. The estimated
posterior marginals for $\beta_2$ (effect of smoking) are shown
in~\Fig{fig:mvprobit1}(a) (solid and dashed lines). The choice of
$\lambda$ seems to have only a minor effect on the posterior marginal
for the effect of smoking $\beta_2$.

Since all the estimated coefficients in $\mm{R}$ are somewhat similar
in the general model for \mm{R} (Figure \ref{fig:corr}(b)), it is
natural to investigate a simplified model with an exchangeable
correlation matrix where all correlations are the same, $\rho$. For
positive definiteness, we require $-1/(m-1) < \rho < 1$.  The fact
that positive and negative correlations are apparently very different,
makes the selection of the prior for $\rho$ challenging.  Due to the
invariance property of the \PCP s, this asymmetry is automatically
accounted for and the potential problem goes away.  We can easily
compute the \PCP\ for $\rho$ for any fixed base value $\rho_0$. For
$\rho_0=0$, which is the same base model as the correlation matrix
\PCP, the distance function to the base model is
\begin{displaymath}
    d(\rho) = \sqrt{
        -\log\left( (1 + (m-1)\rho)(1-\rho)^{m-1}\right)
    }
\end{displaymath}
and the prior follows directly after noting that in this case we must
also allow for $\xi < 0$. The \PCP\ is shown in \Fig{fig:mvprobit1}(b)
for $\lambda=0.1$ (solid) and $1.0$ (dashed). The \PCP\ automatically
adjusts the prior density for $\rho<0$ and $\rho>0$ due to the
constraint $\rho > -1/(m-1)$.

A second potential issue is the following: as we are changing the
model for $\mm{R}$, we should, in order to get comparable results, use
a comparable prior. By reusing the values $\lambda = 0.1$ and $1.0$
from the general case, we define the prior for $\rho$ to penalise the
distance from the base model the same way in both parameterisations of
$\mm{R}$. In this sense, the prior is the same for both models. We can
then conclude that the reduced variability of about $10\%$ for
$\beta_2$ as shown in~\Fig{fig:mvprobit1}(a) for $\lambda=0.1$
(dotted) and $\lambda=1.0$ (dashed-dotted), is due to the more
restrictive exchangeable model for \mm{R} and not an unwanted effect
from the prior distributions for \mm{R}.

The results obtained with the \PCP s are in this example in overall
agreement with those reported in \citet{art570}.

\begin{figure}[tbp]
     \begin{subfigure}{0.45\textwidth}
    \centering
    \includegraphics[width=0.9\textwidth,angle=270]{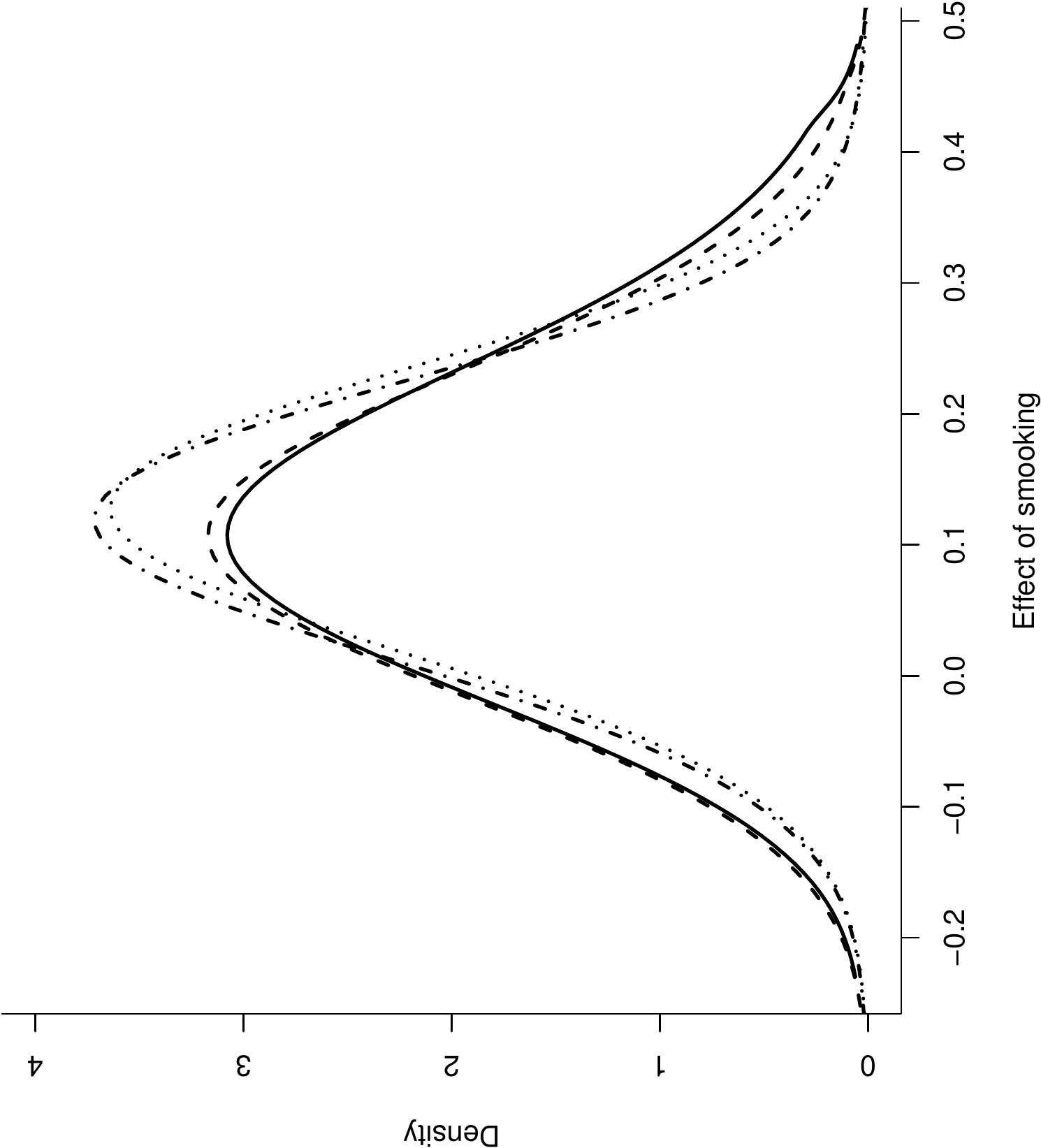}
        \caption{}\label{fig:prec-a}
   \end{subfigure}
         \begin{subfigure}{0.45\textwidth}
    \centering
    \includegraphics[width=0.9\textwidth,angle=270]{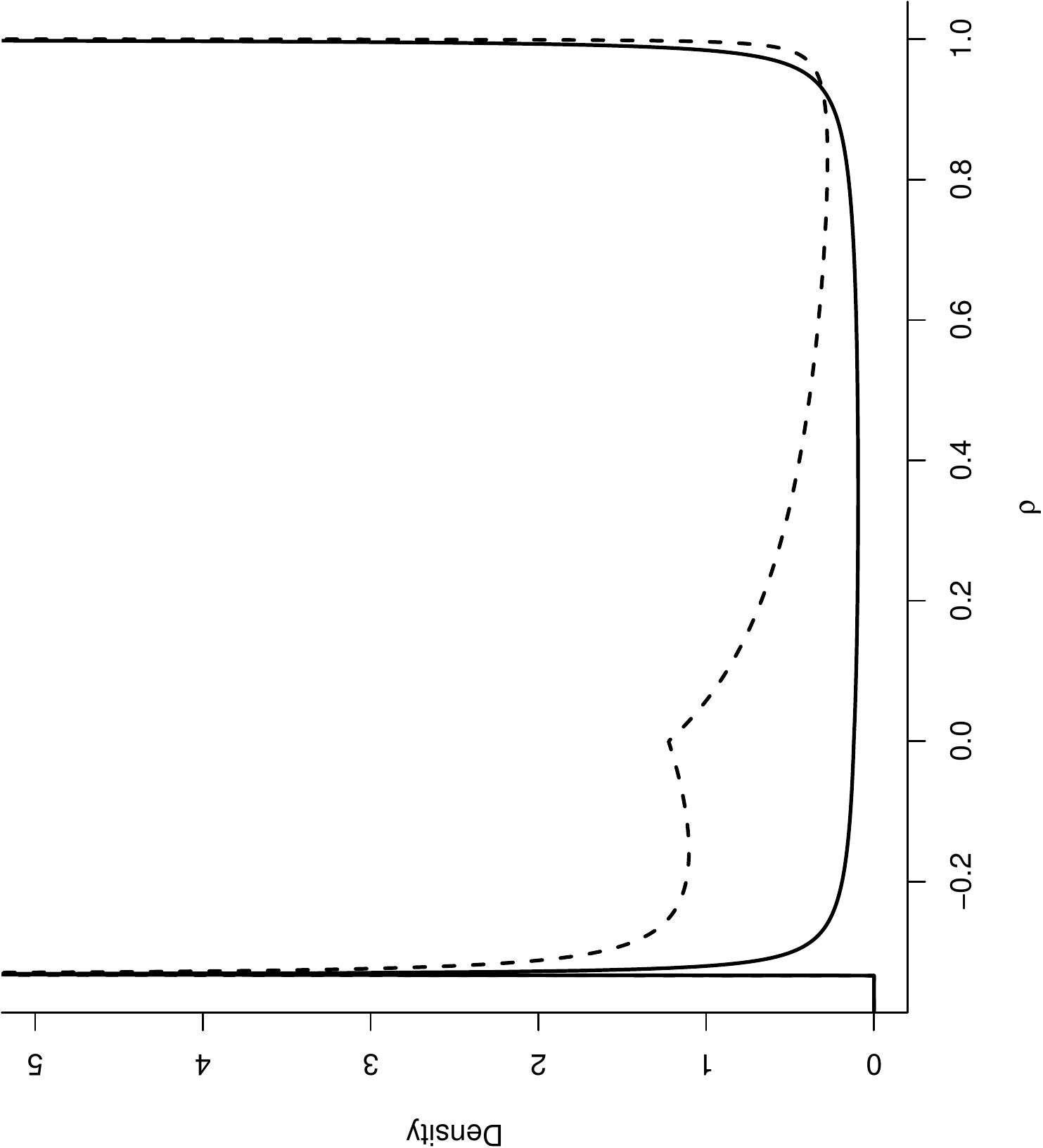}
        \caption{}\label{fig:prec-a}
   \end{subfigure}

    \caption{Panel~(a) shows the estimated posterior marginal for
        $\beta_2$ (the effect of smoking) for both the general model for
        $\mm{R}$ ($\lambda=0.1$ (solid) and $\lambda = 1.0$
        (dashed)) and the exchangeable model ($\lambda=0.1$ (dotted) and $\lambda = 1.0$
        (dashed-dotted)). Panel~(b) shows the \PCP\ for the exchangeable case
        with base-model $\rho=0$, for $\lambda=0.1$ (solid) and $1.0$
        (dashed).}
    \label{fig:mvprobit1}
\end{figure}


\section{Distributing the variance: hierarchical models, and alternative distances}
\label{sec:ext}

For complex models, it is unlikely that practitioners will be able to
provide information about the relative effect of each component in an
hierarchical model. Hence, we can no longer build informative
independent priors on each component but need to consider the global
model structure. Looking back to the example in \Sec{sec:bym_germany},
it is clear from \eref{eq:bym.spline}, that we were not able to
control jointly the variance contribution from the spline and the
combined spatial/random effect term. It could be argued that in these
simple situations, this is less of a problem as priors could easily be
tuned to account for this effect and this strategy is well within
current statistical practice.  In this section we argue and
demonstrate that it is possible to control overall variance
automatically using the \PCP\ framework.  This requires, in concordance with Desideratum D2, that the priors on the individual scaling parameters for each component change as the global model changes.  We will demonstrate this by
considering a logistic regression model with several linear and
non-linear effects, and show how we can take the global structure into
account to control the overall variance of the linear predictor, and
controlling how each term contributes to it. To achieve this, we will
use a multivariate \PCP\ on the fractions of how much each component
contributes to the variance of the linear predictor.

The broader message of this section is that within the \PCP\
framework, it is possible to build priors that respect the global
graphical structure of the underlying model. Additionally, it is
possible to build these priors automatically for new models and
data-sets (which can be integrated into software like
\texttt{R-INLA} or \STAN). The priors depend on the graphical structure and the
model design (or covariate values), but do not, of course, depend on
the observations.  Following this path into the future, we happily
give up global invariance of reparameterisation, as we are adding
another source of information to our prior construction. Additional to
knowledge of the base model and the strength of the penalisation, we
also require expert information about the structure of the model. As
with the previous two information sources, this is not particularly
onerous to elicit.

Our practical approach to handle multivariate \PCP s in this setting
is to use a Taylor expansion around the base model, and approximate it
using a first or second order expansion.  When the base model is an
interior point in the parameter space, then the second order
approximation~\Eref{eq:h.sphere} gives the \PCP\ in \Eref{eq:sphere},
while the linear approximation~\Eref{eq:h.simplex} is appropriate if
the base model is at the border of the parameter space leading to the
\PCP\ in \Eref{eq:simplex}. For the quadratic approximation, it is
well known that
\begin{equation}\label{eq:ext.2}
    \text{KLD} = \frac{1}{2} (\mm{\xi} - \mm{\xi}_0)^{T} I(\mm{\xi}_0)
    (\mm{\xi} - \mm{\xi}_0) + \text{higher order terms}    
\end{equation}
where $\mm{\xi}_0$ is the base model and $I(\mm{\xi}_0)$ is the Fisher
information at the base model. This approximation has resemblance to
the generalisation by~\citet[Sec.~3]{art576} of early ideas by
\citet{book112} for the purpose of hypothesis testing in the the
multivariate case.  \Eref{eq:ext.2} in not unsound as measure of
complexity by itself, and adopting this as our second principle for
practical/computational reasons, then \Eref{eq:sphere} will be the
corresponding \PCP, but will not longer be invariant for
reparameterisation. Hence, care needs to be taken in order to choose a
parameterisation for the second order expansion to be sensible. This
parameterisation is typically motivated by a variance-stabilising
transformation.

To fix ideas, we will discuss a dataset from~\citet{art574} about a
retrospective sample of males in a heart-disease high-risk region of
the Western Cape, South Africa. These data are available as
\texttt{heart} in the \texttt{R}-package \texttt{catdata}, and we will
use the model suggested by~\citet[Sec.~6.4]{art575} changing the link
to logit. A total of 462 subjects are classified of have had ($y_i=1$)
a heart attack or not ($y_i=0$), and the measured risk factors are age
at onset (Age), systolic blood pressure (BP) and low density
lipoprotein cholesterol (CR). We use standardised covariates in this
analysis.  The linear predictor is
\begin{displaymath}
    \mm{\eta} = \mu \mm{1} + \tau^{-1/2}\times g\left(
      \text{Age}, \text{BP}, \text{CR}\right)
\end{displaymath}
where $g(\cdot)$ is some smooth function of the covariates. At this
top-level, we can use the structural information provided by the model
to elicit the amount of variability we expect from covariates. This
information can be incorporated in the prior for the precision
parameter $\tau$.  We assume here that the effect of covariates
$g(\cdot)$ have zero mean and ``unit variance''. We use the phrase
``unit variance'' for $\beta_x \mm{x}$ to describe a standardised
covariate $\mm{x}$ and a covariate weight $\beta_x$ with unit
variance. The predicted probabilities from the regression model might
mostly be in the range $[0.05, 0.95]$ leading to an interval
$[-2.94,2.94]$ on the linear predictor scale. We take the marginal
standard deviation of the effect of the covariates to be $2.94/1.96$,
which gives parameters $U=4.84$ and $\alpha=1\%$ in the \PCP\ for the
precision in a Gaussian random effect~\eref{eq:prec}. This prior will
shrink the effect of the covariates towards the intercept, which is
the first base model.

At the next level in the model, we use an additive model for the
covariates and distribute the unit variance among the covariates,
\begin{displaymath}
    g(\text{Age}, \text{BP}, \text{CR}) = \sqrt{w_1} g_1(\text{Age}) +
    \sqrt{w_2} g_2(\text{BP}) + \sqrt{w_3} g_3(\text{CR})
\end{displaymath}
where the weights live on a 2-simplex, i.e.\ $w_1+w_2+w_3=1$ and $w_i
\ge 0$, and $\{g_i(\cdot)\}$ are smooth functions (with unit
variance). The variance contribution from covariate Age, say, is then
$w_1$. Without additional knowledge, it is reasonable to treat them
equally, meaning that the base model for the weights is $w_1 = w_2 =
w_3 = 1/3$.  This reflects an \emph{a priori} assumption of
(conditional) exchangeability between these model components.

Further one level down, we continue to distribute the variance, but
now for each $g_i(\cdot)$ function and between a linear and (purely)
non-linear effect. For covariate Age, this becomes
\begin{displaymath}
    g_1(\text{Age}) = \sqrt{1-\phi_{1}} \beta_1 \text{Age} +
    \sqrt{\phi_1} f_1(\text{Age}), \quad \phi_1 \ge 0.
\end{displaymath}
Here, both $\beta_1$ and $f_1(\cdot)$ have unit variance, and
$f_1(\cdot)$ is a smooth (purely) non-linear function. The natural
base model is $\phi_1=0$ meaning that the variance in
$g_1(\text{Age})$ is only explained by the linear effect, as we do not
want to involve deviations from the linear effect without support from
data.

\begin{figure}
    \begin{subfigure}{0.6\textwidth}
        \centering
        \resizebox{\textwidth}{!}{\begin{tikzpicture}
    \tikzstyle{latent} = [circle,fill=white, draw=black, inner sep=1pt,
    minimum size=30pt, font=\fontsize{10}{10}\selectfont, node distance=1]
    
    \node[latent] (beta1) {$\beta_1\times\text{\small Age}$};
    \node[latent, right=of beta1] (f1) {$f_1(\text{\small Age})$};
    \node[latent, above=of beta1, xshift=1.15cm] (g1) {$g_1(\text{\small Age})$};
    \factor[above=of beta1,xshift=5mm]{beta1-g1}{left:${1-\phi_1}$} {} {};
    \factor[above=of f1,xshift=2mm]{f1-g1}{left:${\phi_1}$} {} {};
    \edge{g1}{beta1}
    \edge{g1}{f1}

    \node[latent, right=of f1] (beta2) {$\beta_2\times\text{\small BP}$};
    \node[latent, right=of beta2] (f2) {$f_2(\text{\small BP})$};
    \node[latent, above=of beta2, xshift=1.15cm] (g2) {$g_2(\text{\small BP})$};
    \factor[above=of beta2,xshift=5mm]{beta2-g2}{left:${1-\phi_2}$} {} {};
    \factor[above=of f2,xshift=2mm]{f2-g2}{left:${\phi_2}$} {} {};
    \edge{g2}{beta2}
    \edge{g2}{f2}

    \node[latent, right=of f2] (beta3) {$\beta_3 \times\text{\small CR}$};
    \node[latent, right=of beta3] (f3) {$f_3(\text{\small CR})$};
    \node[latent, above=of beta3, xshift=1.15cm] (g3) {$g_3(\text{\small CR})$};
    \factor[above=of beta3,xshift=5mm]{beta3-g3}{left:${1-\phi_3}$} {} {};
    \factor[above=of f3,xshift=2mm]{f3-g3}{left:${\phi_3}$} {} {};
    \edge{g3}{beta3}
    \edge{g3}{f3}

    \node[latent, above=of g2] (g) {$g(\text{\small Age,BP,CR})$};
    \edge{g}{g1}
    \edge{g}{g2}
    \edge{g}{g3}
    \factor[above=of g1, xshift=17mm]{g1-g}{left:${w_1}$} {} {};
    \factor[above=of g2, xshift=0mm]{g2-g}{left:${w_2}$} {} {};
    \factor[above=of g3, xshift=-7mm]{g3-g}{left:${w_3}$} {} {};


\end{tikzpicture}


}
        \caption{}
    \end{subfigure}
    \hspace{0.\textwidth}
    \begin{subfigure}{0.3\textwidth}
        \centering
        \rotatebox{270}{\includegraphics[width=\textwidth]{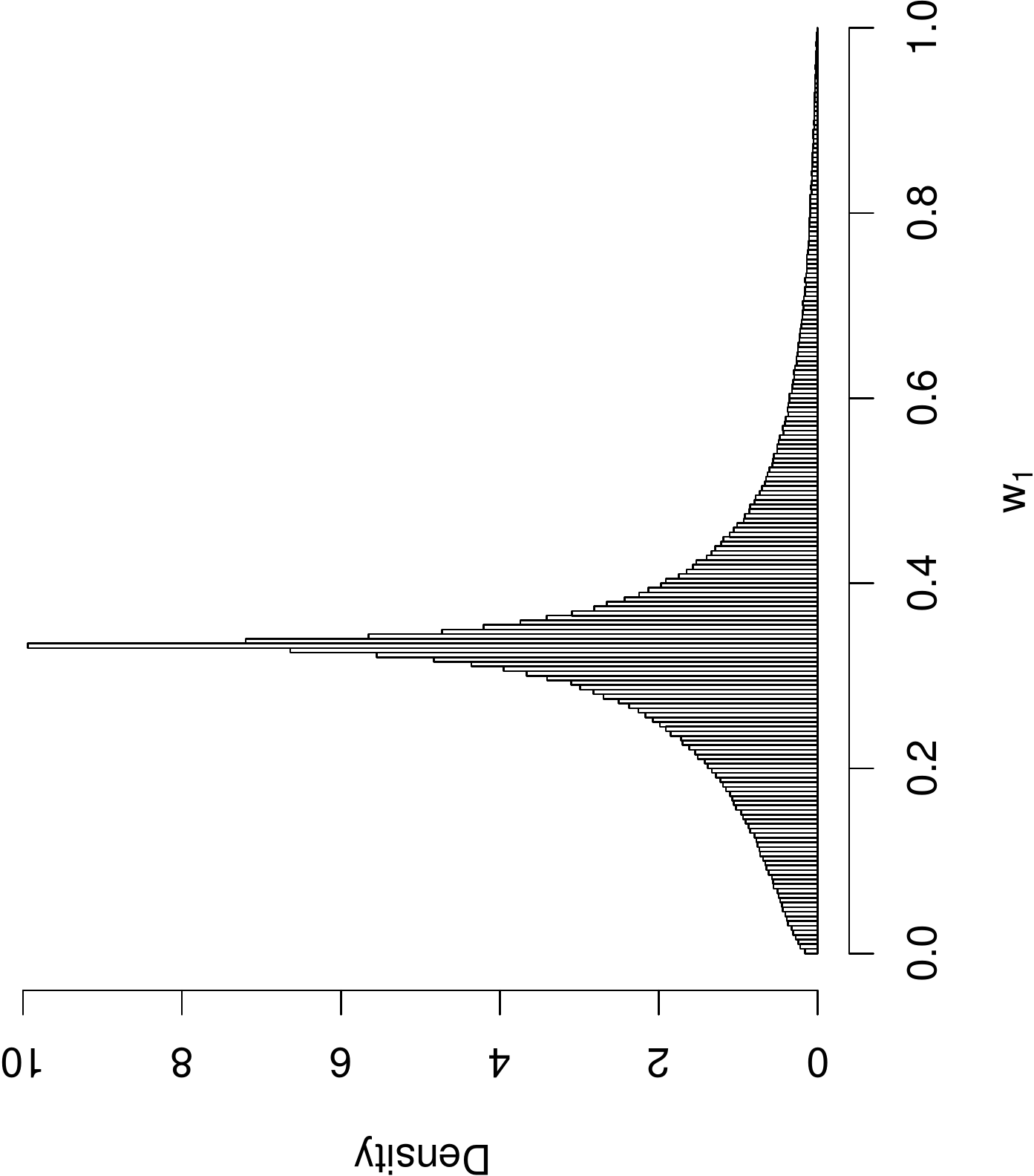}}
        \caption{}
    \end{subfigure}
    \caption{Panel~(a) displays the recursive structure of how the
        additive effect of the covariates are built up, to control
        the variance contributed from each of the model terms to the
        linear predictor. Panel~(b) shows the histogram of $w_1$ from
        a typical joint \PCP\ for $\mm{w}$ using $\lambda = 0.3$.  }
    \label{fig:pc.variance}
\end{figure}
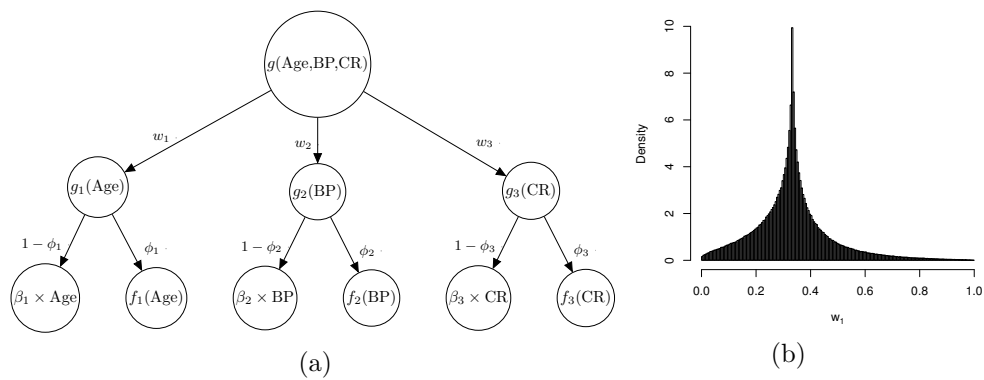
\Fig{fig:pc.variance}~(a) displays the graphical structure of the
recursive decomposition of the variance of $g(\text{Age}, \text{BP},
\text{CR})$. By following the path from the top node to the relevant
node we can determine the fraction of the variance explained by that
node.  For example, the relative contribution to the variance from the
linear effect of covariate Age, is $w_1(1-\phi_1)$, and the relative
contribution to the variance from $g_3(\text{CR})$ is $w_3$.

In order to proceed with the analysis and computation of the \PCP, we
need to make some specific choices for the linear effects
$\{\beta_i\}$ and the (purely) non-linear effects
$\{f_i(\cdot)\}$. For $\beta_i$ we use independent zero mean Gaussians
with unit variance, and for the non-linear effect, we use the second
order random walk~\eref{eq:rw2} which corresponds to the integrated
Wiener process used by~\citet{art575}; see \citet{art435} for a
discussion of this issue. \Fig{fig:pc.variance}~(b) displays the
histogram of samples for the first weight component $w_1$ derived
using a typical \PCP\ for the weights $\mm{w}$ with $\lambda = 0.3$.
The histogram is centred at the base model $w_1 = 1/3$, but still
supports weights close to $0$ meaning that the covariate Age does not
contribute to the linear predictor, and close to 1 meaning that
\emph{only} Age contributes to the linear predictor.

The \PCP\ in this example is computed as follows. The priors for
$\mm{\phi}$ follows the computations described
in \Appendix{appendix:bym}. The joint prior for $\mm{w}$, depends on
the values for $\mm{\phi}$ (but not too much in this example), and
therefore has to be recomputed when there is any change in
$\mm{\phi}$. We use here~\eref{eq:sphere} as an approximation to the
multivariate \PCP\ which only requires a numerical estimate of the
Hessian matrix at the base model. More details are available
in \Appendix{sec:ext.appendix}. The results were obtained using
\texttt{R-INLA}. The covariate estimates (not shown) are comparable to
those in \citet{art575} obtained using MCMC based on independent
diffuse and flat priors.


\section{Discussion}
\label{sec:disc}

Prior selection is \emph{the} fundamental issue in Bayesian
statistics.  Priors are the Bayesian's greatest tool, but they are
also the greatest point for criticism: the arbitrariness of prior
selection procedures and the lack of realistic sensitivity analysis
(which is addressed in \citet{roos-held-2011} and \citet{roos2013})
are serious arguments against current Bayesian practice.  In this
paper, we have provided a principled, widely applicable method for
specifying priors on parameters that are difficult to directly elicit
from expert knowledge.  These \PCP s can be vague, weakly informative,
or strongly informative depending on the way the user tunes an
intuitive scaling parameter.  The key feature of these priors is that
they explicitly lay out the assumptions underlying them and, as such,
these assumptions and principles can be directly critiqued and
accordingly modified.

\PCP s are defined on individual components.  This distinguishes \PCP
s, from reference priors, in which the priors depend on the global
model structure. This global dependence is required to ensure a proper
posterior. However, the modern applied Bayesian is far more likely to
approach their modelling using a component-wise and often additive
approach. The directed acyclic graph--approach pioneered by the
WinBUGS inference engine, is now a standard tool for specification of
general Bayesian models. The additive approach pioneered
by~\citet{book72} is now a standard approach within generalised
regression based models.  Hence, the ability to specify priors in a
component-wise manner is a useful feature.  It is worth noting that
 none of the examples in this paper  have known
reference priors. We believe that \PCP s are a valuable addition to
the literature on prior choice.  They are not designed as, and should
not be interpreted as, a replacement for reference priors, but rather
a method to solve a different set of problems.

This is not the whole story of \PCP s. We still have to work them out
on a case by case basis, construct better guidance for choosing the
scaling using knowledge of the global model (like the link-function
and the likelihood family), and make them the default choice in
packages like \RINLA. Not all cases are straight forward.  The
over-dispersion parameter in the negative binomial distribution,
considered as an extension of the Poisson distribution, cannot be
separated from the mean. Hence, we cannot compute the \PCP\ without
knowing a typical value for the mean. We also need to get more
experience deriving joint priors for two or more parameters, such as a
joint prior for the skewness and kurtosis deviation from a Gaussian
\citep{art566}.  We also have not considered \PCP s on discrete
parameters.  To do this, we need to find a sensible, computationally tractable notion of distance
for these problems.  In this paper, we have focused on generic
specification, however \citet{fuglstad2015interpretable} show that, in the case of hyperparameters for Gaussian random fields, 
if the distance knows about the structure of the model component, the
resulting priors perform very well.  Hence, there is more work to be
done on tailoring distance measures to specific problem classes.

Several of the examples in this paper have used the notion that model
components can be regarded as a flexible extension of a base model.
This idea has natural links to Bayesian non-parametrics.  In
particular, we consider many of the examples, such as the logistic GAM
model in Section \ref{sec:ext}, as a non-parametric model that has
been firmly rooted in a simpler parametric model.  We believe that
this decomposition of the model into favoured parametric and extra
non-parametric components improves the interpretability of these
models in many applications.  This is related to the ideas of
\citet{Kennedy2001}, where the nonparametric component can be used to
``calibrate'' the covariate effect.  An alternative interpretation of
this decomposition is that the nonparametric part adds ``robustness''
to the linear model and the flexibility parameter gives an indication
of how far from the simple, interpretable, base model the data
requires you to depart.

There is a great deal of scope for further theoretical work on this
problem.  Firstly, it would be useful to understand better the effect
of the prior tail on the inference.  The results in this paper suggest
that an exponential tail is sufficiently heavy in low-dimensional
problems, and the heavy tailed half-Cauchy distribution only gives
different results in the high dimensional sparse setting.  However,
it's not clear that this is truly a problem with the tail, as an
examination of the base model suggests that it is not shrinking
towards sparse models.  Hence the question is ``are heavy tails
necessary in high dimensions, or are they just more forgiving of our
poor prior specification?''.  Staying in the realm of sparse models,
it would be interesting to see if the models in \Sec{sec:ext} could be
extended to high dimensional sparse models.  It may be possible to
take inspiration in this case from the Dirichlet--Laplace construction
of \citet{bhattacharya2012bayesian}.  More generally, there are open questions
relating to the large sample properties of hierarchical models with PC priors, hypothesis testing
for flexible models, Stein-type risk properties for PC priors, and robustness against mis-specification.

The current practice of prior specification is not in a good shape.
While there has been a strong growth of Bayesian analysis in science,
the research field of ``practical prior specification'' has been left
behind.  There are few widely applicable guidelines for how this could
or should be done in practice.  We hope that with this new principled
approach to prior construction, we can reduce the number of ``cut and
paste'' prior choices from other similar research articles, and
instead use the derived tools in this paper to specify weakly
informative or informative priors with a well defined shrinkage. As
always, if the user knows of a better prior for their case, then they
should use it. However, having a better default proposal for how to
construct priors is a significant advantage.  The \PCP\ framework was
constructed because of our work with  scientists on applied
problems and came out of a desire to derive and explain the prior
information that we were putting into hierarchical models.  As such,
we believe that these priors are ``fit for purpose'' as tool for
real-world applied statistics.

These new \PCP\ have made a difference to how we do and see Bayesian
analysis. We feel much more confident that the priors we are using do
not force over-fitting, and the notion of \emph{scale}, which
determines the magnitude of the effects, really simplifies the
interpretation of the results. The fact that the prior specification
reduces to a notion of \emph{scale}, makes them very easy to interpret
and communicate.  We feel that \PCP s lay out a new route forward
towards more sound Bayesian analysis.  Jimmie Savage \citep{Lindley1983}
suggested that we ``build our models as big as
elephants'', while J. Bertrand \citep{le1990maximum} told us to ``give [him] four parameters and [he] shall describe an elephant; with five, it will wave its trunk''.  The modern practice of Bayesian statistics can be seen as a battle between these two elephants, and with \PCP s we hope to make it a fair fight.


\section*{Acknowledgement}

The authors acknowledge
Gianluca Baio,
Haakon C.\ Bakka,
Simon Barthelm\'e,
Joris Bierkens,
Sylvia Fr\"uhwirth-Schnatter,
Nadja Klein,
Thomas Kneib,
Alex Lenkoski, 
Finn K.\ Lindgren,
Christian P.\ Robert and
Malgorzata Roos
for stimulating discussions and comments related to this work.

\appendix
\section{Derivation of \PCP s}

\subsection{The ``distance'' to a singular model}

As \PCP s are defined relative to a base model, which is essentially a
distinguished point in the parameter space, we occasionally run into
the difficulty that this point (denoted $\xi = 0$) is fundamentally
different from the points $\xi > 0$.  In particular, the base model is
occasionally singular to the other distributions and we need to define
a useful notion of distance from a singular point in the parameter
space.  We are saved in the context of this paper by noting that the
singular model $\pi(x|\xi=0)$ is the end point of a curve in model
space $t \rightarrow \pi(x | \xi = t)$, $t\geq0$.

Consider the $\epsilon$-distance $d_\epsilon(t) = \sqrt{2 \KLN{\pi(x |
        \xi = t)}{\pi(x | \xi=\epsilon)}}$, which is finite for every
$\epsilon >0$. If $d_\epsilon(t) = \mathcal{O}(\epsilon^{-p})$, then
we can define the renormalised distance $\tilde{d}_\epsilon(t) =
\epsilon^{p/2}d_\epsilon(t)$.  Using this renormalisation,
$\tilde{d}(t) = \lim_{t\downarrow 0}\tilde{d}_\epsilon(t)$ is finite
and parameterisation invariant and we can use it to define \PCP s for
singular models.

\subsection{The \PCP\ for the precision in a multivariate Normal
    distribution}
\label{appendix:normal}

Let ${\mathcal N}^{(p)}(\mm{\mu}, \mm{\Sigma})$ denote a multivariate
normal distribution with dimension $p$.  The Kullback-Leibler
divergence from ${\mathcal N}_1^{(p)}(\mm{\mu}_1, \mm{\Sigma}_1)$ to
${\mathcal N}_0^{(p)}(\mm{\mu}_0, \mm{\Sigma}_0)$ is
\begin{equation}\label{appendix:kld}%
    \KLN{{\mathcal N}_1^{(p)}}{{\mathcal N}_0^{(p)}} =
    \frac{1}{2} \left\{
      \text{tr}\left( \mm{\Sigma}_0^{-1} \mm{\Sigma}_1\right)
      + (\mm{\mu}_0 - \mm{\mu}_1)^{T} \mm{\Sigma}_0^{-1}
      (\mm{\mu}_0 - \mm{\mu}_1)
      - p
      -\ln\left( \frac{|\mm{\Sigma}_1|}{|\mm{\Sigma}_0|}\right)
    \right\}.
\end{equation}
In our setting, ${\mathcal N}_1^{(p)}$ denotes the flexible model and
${\mathcal N}_0^{(p)}$ the base model.

To derive the \PCP\  for $\tau$ where $\mm{\Sigma}_1 =
\mm{R}/\tau$, $\mm{R}$ is a fixed matrix and $\mm{\Sigma}_0 =
\mm{0}$, we will study the limiting behaviour when $\mm{\Sigma}_0 =
\mm{R}/\tau_0$ for a high fixed value of $\tau_0$. In the end we
will look at the limit $\tau_0 \rightarrow \infty$. For simplicity,
assume $\mm{R}$ has full rank. We then get $
    \text{KLD} = \frac{p}{2} \frac{\tau_0}{\tau} \left(
      1 + \frac{\tau}{\tau_0} \ln\left(\frac{\tau}{\tau_0}\right) -
      \frac{\tau}{\tau_0}
    \right) \longrightarrow \frac{p}{2} \frac{\tau_0}{\tau}$, when $\tau \ll \tau_0$,
and hence $d(\tau) = \sqrt{p\tau_0/\tau}$. With an exponential prior
for $d$ with rate $\lambda = \theta/\sqrt{p\tau_0}$ we get they type-2 Gumbel distribution \eqref{eq:prec}. To infer $\theta$ from a notion
of \emph{scale}, we request $(U, \alpha)$ saying that large standard
deviations are less likely, $\Prob\left({1}/{\sqrt{\tau}} > U
\right) = \alpha$.  This gives $\theta = -\ln(\alpha) / U$.  We can
now take the limit $\tau_0 \rightarrow \infty$ by choosing $\lambda$
so that $\theta$ is kept constant. In the case where $\mm{R}$ is not
of full rank, we use the generalised inverse and determinant, and we
arrive at the same result.

\subsection{The \PCP\  for the mixing parameter in the
    BYM-model}
\label{appendix:bym}

We will now derive the \PCP\  for the mixing parameter $\phi$ in
the new parameterisation for the BYM model.  Let $\mm{u}$ be a
$n$-dimensional standardised Gaussian model with zero mean and
precision matrix $\mm{R}>0$, $\mm{v}$ be an independent zero mean
random effects with unit variance ${\mathcal N}(\mm{0}, \mm{I})$, and
where the mixing parameter $\phi$ satisfies $0 \le \phi \le 1$.  The
more flexible model is $\sqrt{1-\phi}\; \mm{v} + \sqrt{\phi}\;
\mm{u}$, and the base model is $\mm{v}$ (i.e.\ the model flexible
model when $\phi=0$). Let $\mm{\Sigma}_0 = \mm{I}$ and
$\mm{\Sigma}_1(\phi) = (1-\phi) \mm{I} + \phi \mm{R}^{-1}$, then
\begin{eqnarray*}
    2\,\text{KLD}(\phi) &=& 
    \text{tr}(\mm{\Sigma}_1(\phi)) -n -\ln
    |\mm{\Sigma}_1(\phi)| \\
    &=&
    n\phi\left(\frac{1}{n}\text{tr}(\mm{R}^{-1}) - 1\right)
    -
    \ln | (1-\phi) \mm{I} + \phi \mm{R}^{-1}|
\end{eqnarray*}
and $d(\phi) = \sqrt{2\text{KLD}(\phi)}$.  The interesting case is
when $\mm{R}$ is sparse, for which $\text{tr}(\mm{R}^{-1})$ is quick
to compute~\citep{art358,art375}. For the determinant term, we can
massage the expression to facilitate the speedup of computing with
sparse matrices. Using the matrix identity $(\mm{I} +
\mm{A}^{-1})^{-1} = \mm{A}(\mm{A} + \mm{I})^{-1}$, we get
$
    | (1-\phi) \mm{I} + \phi \mm{R}^{-1}| =
     | \phi^{-1}\mm{R} |^{-1}{|\frac{1-\phi}{\phi} \mm{R} + \mm{I} |}{
        }$.
An alternative approach, is to compute the eigenvalues $\{\gamma_i\}$
of $\mm{R}$, which we need to do only once.  Let $\tilde{\gamma}_i =
1/\gamma_i$, and we get $
    | (1-\phi) \mm{I} +\phi \mm{R}^{-1}| = \prod_{i=1}^{n}
    \left(
      1-\phi + \phi \tilde{\gamma}_i
    \right)$.
    
In the case where $\mm{R}$ is singular we introduce linear
constraint(s) to ensure that any realisation of $\mm{u}$ is in its
null-space. It is now easier to use the latter computational strategy,
but redefine $\tilde{\gamma}_i$ as $1/\gamma_i$ if $\gamma_i > 0$ and
$\tilde{\gamma}_i = 0$ if $\gamma_i = 0$.

The truncated exponential prior for $d(\phi)$ gives
\begin{displaymath}
    \pi(\phi) = \frac{\lambda \exp\left(-\lambda d(\phi)\right)}{
        1-\exp\left(-\lambda d(1)\right)}
    \; \bigg\vert\frac{\partial d(\phi)}{\partial\phi}\bigg\vert
\end{displaymath}
where we do not spell out the Jacobian for simplicity.  The
penalisation parameter $\lambda$ can be determined by $\Prob(\phi < u)
= \alpha$, requiring that $\alpha > d(u)/d(1)$.

The \PCP\ for the mixing parameter in the hierarchical models in 
\Sec{sec:ext} generalise the BYM-model as the base model is
more general.  Let $\mm{\Sigma}_1(\phi) = (1-\phi) \mm{S}_1 + \phi
\mm{S}_2$, and where the base model is $\mm{\Sigma}_0 =
\mm{\Sigma}_1(1/2)$. The costly task is to compute
$\det(\mm{\Sigma}_1(\phi))$ for a sequence of $\phi$'s. Using the
Matrix determinant lemma: $\det(\mm{A} + \mm{U}\mm{V}^{T}) =
\det(\mm{I} + \mm{V}^{T}\mm{A}^{-1}\mm{U})\det(\mm{A})$ for compatible
matrices $\mm{A}$ (invertible), $\mm{V}$ and $\mm{U}$, we can reduce
the computational cost to essential one evaluation of
$\det(\mm{\Sigma}_1(\phi))$.
\subsection{The \PCP\  for the variance weights in additive models}
\label{sec:ext.appendix}

The joint \PCP\ of the weights \mm{w} in \Sec{sec:ext} is computed as
follows. Let $\mm{\eta}^{*}$ be the standardised linear predictor and
$\mm{x}_i$ the $i$'th vector of standardised covariates, then the model
considered in \Sec{sec:ext} can be written as
$
    \mm{\eta}^{*} = 
    \sum_i \sqrt{w_i} \left(
      \sqrt{1-\phi_i} \beta_i \mm{x}_i + \sqrt{\phi_i} \mm{A}_i \mm{f}_i
    \right)$,
where $\mm{A}_i$ is a sparse matrix extracting the required elements
(or linear combinations thereof) of the Gaussian vector $\mm{f}_i$
representing the scaled second order random walk model. The covariance
for the linear predictor is then
$
    \text{Cov}(\mm{\eta}^{*}) = \sum_i
    w_i\left(
      (1-\phi_i) \mm{x}_i\mm{x}_i^{T} + \phi_i \mm{A}_i
      \text{Cov}(\mm{f}_i) \mm{A}_i^{T}
    \right)$.
In order to improve the second order approximation~\Eref{eq:sphere},
we reparameterise the weights following the ideas in compositional
data analysis~\citep{book122}, using new parameters $\tilde{w}_i =
\log(w_i/w_n)$, for $i=1, \ldots, n-1$ for $n$ components. This makes
$\text{Cov}(\mm{\eta}^{*})$ a function of $\tilde{\mm{w}}$ with base
model at $\tilde{\mm{w}} = \mm{0}$. The KLD can then be computed
from~\Eref{appendix:kld}, and the \PCP\ follows from a numerical
approximation to the Hessian matrix of the KLD and \Eref{eq:sphere}.


\section{Proofs}
\subsection{Proofs of Theorems \ref{theorem:2} and \ref{theorem:1}}
\label{sec:proof}
We give the proof of \Thm{theorem:1}.  The proof of \Thm{theorem:2} follows along the same lines.
The KLD of approximating the unit precision Student-t with
d.o.f.~$\nu$ with a standard Gaussian is for large $\nu$, $
    \text{KLD} = \frac{3}{4}\nu^{-2} + \frac{3}{2}\nu^{-3} +
    {\mathcal O}(\nu^{-4})$.
Since $\pi_{\nu}(\nu)$ has a finite first moment it is $o(\nu^{-2})$
as $\nu \rightarrow \infty$. Using the fact that $d =
\sqrt{2\;\text{KLD}}$, shows that $\pi_d(d) = o(1)$ as $d\rightarrow
0$. 

\subsection{Proof of \Thm{thm:BF}}
\label{appendix:BF}
Let $\mm{y}_n$ be $n$ i.i.d.\ draws from $\pi(y|\zeta)$ and consider
the hypothesis test $H_0: \zeta=0$ against the alternative $H_1: \zeta
\sim \pi(\zeta)$, where $\pi(0)\in(0,\infty)$.  Let $m_i(\mm{y}_n)$
denote the marginal likelihood under each model.  As the domain of
$H_1$ is open and contains only regular points for the model, the
consistency of $B_{01} $ under $H_1$ follows from
\citet{johnson2010use}. Assume that the model has regular asymptotic
behaviour under $H_0$.  Under $H_0$, Bayes' theorem implies that, for
any $\zeta \geq 0$,
\begin{align*}
    B_{01}(\mm{y}_n)=\frac{m_0(\mm{y}_n) }{m_1(\mm{y}_n)} &=
    \frac{\pi(\mm{y}_n\mid \zeta =0)}{\pi(\mm{y}_n \mid \zeta)}
    \frac{\pi(\zeta \mid \mm{y}_n)}{\pi(\zeta)}
    \stackrel{p}{\longrightarrow} \frac{\pi(\mm{y}_n\mid \zeta
        =0)}{\pi(\mm{y}_n \mid \zeta)}\frac{\exp(-v/2)}{\pi(0)}
    \sqrt{\frac{n}{8\pi v}},
\end{align*}
where the second equality follows from
\citet[Thm.~1]{bochkina2012bernstein}, which states that 
$\pi(\zeta|\mm{y}_n)$ converges to a truncated normal distribution
with mean parameter $0$ and variance parameter $n^{-1/2}v$, and we
have evaluated this asymptotic density at $\zeta =n^{-1/2}$. That
$B_{01}(\mm{y}_n) = \mathcal{O}_p(n^{1/2})$ follows by noting that
\citet[Assumption M1]{bochkina2012bernstein} implies that the first
quotient is $\mathcal{O}_p(1)$.

When the model has irregular asymptotic behaviour, the result follows
by replacing the truncated normal distribution by the appropriate
Gamma distribution.  It follows that $B_{01}(\mm{y}_n) =
\mathcal{O}_p(n)$ in this case. The results of
\citet{bochkina2012bernstein} can also be used to extend this to the
parameter invariant extensions of the non-local priors considered by
\citet{johnson2010use}.  If $\pi(\zeta) = \mathcal{O}(\zeta^k)$,
$k>-1$ as $\zeta \rightarrow 0$, then a simple extension of the above
argument shows that $B_{01}(\mm{y}_n) = \mathcal{O}_p(n^{1/2+k})$ in
the regular case and $B_{01}(\mm{y}_n) = \mathcal{O}_p(n^{1+k})$ in
the irregular case.

\subsection{Proof of \Thm{thm:shrinkage}}

The theorem follows by noting that as $\kappa\downarrow 0$,
$\pi(\kappa) =
\mathcal{O}\left(\pi_d\left(\kappa^{-1/2}\right)\kappa^{-3/2}\right)$
and as $\kappa \uparrow 1$, $\pi(\kappa) =
\mathcal{O}\left(\pi_d\left(\sqrt{1-\kappa}\right)(1-\kappa)^{-1/2}\right)$.

\subsection{Proof of \Thm{thm:expon_decay}}
\label{appendix:normal_proof}

For convenience, we will prove Theorem 4 for
general priors with $\pi_d(0) \in (0, \infty)$.  Consider a prior
$\pi(\sigma)$ on the standard deviation with $\pi(0) = 1$ and define
the re-scaled prior as $\pi^\lambda(\sigma) =
\lambda\pi(\lambda\sigma)$.  The following theorem shows that, for
sufficiently large $\lambda = \lambda(p)\uparrow \infty$, the marginal
prior for $\mm{\beta}$ has mass on $\delta$--sparse vectors.

\begin{theorem}
    \label{thm:general_sparse}
    Let $\pi(\sigma)$ be a non-increasing prior on $\sigma$ such that
    $\pi(0)=1$ and let $\mm{\beta} \sim {\mathcal N}(\mm{0},\mm{D}^2)$, where
    $D_{ii} \sim \pi^\lambda(\sigma)$.  Set $\delta_p = p^{-1}$.  
    Sufficient condition for the prior on the $\delta_p$--dimension to
    be  centred at the true sparsity $s$ are that $\lambda \geq
    \mathcal{O}\left(\frac{p}{ \log
          (p)}\left[1-\frac{s}{p}\right]\right)$ and
    $\pi^\lambda(p^{-1}) \leq \mathcal{O}\left(\frac{p}{ \log
          (p)}\left[1-\frac{s}{p}\right]\right)$, where $s$ is the
    true sparsity of the target vector.
\end{theorem}
\begin{proof}
    The result follows by noting that
    \begin{align*}
        1-{\alpha} &= 2\int_0^\delta\int_0^\infty
        (2\pi\sigma^2)^{-1/2}\exp\left(\frac{\beta^2}{2\sigma^2}\right)\pi^\lambda(\sigma)\,d\sigma\,d\beta \\
         &\gtrsim \lambda \pi(\lambda\delta)
        \int_0^\delta \int_0^1 \sigma^{-1} \exp\left(\frac{\beta^2}{2\sigma^2}\right)\,d\lambda\,d\beta \\
        &\gtrsim \lambda \pi(\lambda\delta)\int_0^\delta
        \log\left(1+\frac{4}{\beta^2}\right)e^{-\frac{\beta^2}{2}}\,d\beta
        \gtrsim \lambda\pi(\lambda\delta)
        \log\left(1+\frac{4}{\delta^2}\right)\operatorname{erf}(2^{-1/2}\delta) \\
       & \gtrsim \lambda\pi(\lambda\delta)\delta \log (\delta)^{-1},
    \end{align*}
    where the first inequality comes from the definition of
    $\pi^\lambda(\sigma)$ and the second follows from standard bounds
    on the exponential integral.  Noting that $$\int_1^\infty
    (2\pi\sigma^2)^{-1/2}
    \exp\left(\frac{\beta^2}{2\sigma^2}\right)\pi^\lambda(\sigma)\,d\sigma
    \lesssim 1$$ and $\pi^\lambda(\sigma) \lesssim 1$, a similar
    calculations yield $1-\alpha \lesssim \lambda \delta
    \log(\delta^{-1})$.  It follows that $\alpha = p^{-1}s$
    when $\lambda\pi\left(\frac{\lambda }{p}\right)\lesssim
    \frac{p}{\log(p)} \left(1-\frac{s}{p}\right) \lesssim \lambda$,
    which implies the result.
\end{proof}

Theorem 4 follows from the assumption that $s \leq
\mathcal{O}\left(\frac{p}{\log(p)}\right)$ and using the Taylor
expansion on the Lambert-W function to get the upper bound.

\section{Supplementary material}

The \texttt{R}-code for analysing all examples and generating the
corresponding figures in this report, is available at\newline
{\small\texttt{www.r-inla.org/examples/case-studies/pc-priors-martins-et-al-2014}}

{\bibliography{../mybib,../local-bibfile}}

\end{document}